\documentclass[sigconf]{acmart}
\pdfoutput=1

\def\BibTeX{{\rm B\kern-.05em{\sc i\kern-.025em b}\kern-.08em
    T\kern-.1667em\lower.7ex\hbox{E}\kern-.125emX}}

\usepackage{amsfonts}

\usepackage{textcomp}

\usepackage{hyperref}
\usepackage{graphicx}
\usepackage{alltt}
\usepackage{multicol} 
\usepackage{multirow}
\usepackage{xcolor}
\usepackage{balance}
\usepackage{amsmath}
\usepackage{amsthm}
\usepackage{algorithm,algorithmic}
\usepackage{mathtools}
\usepackage{tikz}
\usepackage{subcaption}
\usepackage{setspace}
\usepackage{boxedminipage}
\usepackage{array}
\usepackage{makecell}
\usepackage{enumitem}
\usepackage{xpatch}

\newfloat{algorithm}{t}{lop}

\newtheorem{theorem}{Theorem}

\renewcommand{\algorithmicrequire}{\textbf{Input:}}
\renewcommand{\algorithmicensure}{\textbf{Output:}}

\newcommand{\dpsub}{\tt DPSUB}
\newcommand{\dpsize}{\tt DPSIZE}
\newcommand{\dpccp}{\tt DPCCP}

\newcommand{\dpe}{\tt DPE}
\newcommand{\mpdp}{\tt MPDP}
\newcommand{\mpdptree}{\tt MPDP:Tree}
\newcommand{\pdp}{\tt PDP}
\newcommand{\lindp}{\tt LinDP}
\newcommand{\innerCounter}{\tt EvaluatedCounter}
\newcommand{\ccpCount}{\tt CCP-Counter}
\newcommand{\ccpPair}{\tt CCP-Pair}
\newcommand{\ccpPairs}{\tt CCP-Pairs}
\newcommand{\genCount}{\tt Generated Counter}
\newcommand{\joinpair}{\text{\sf Join-Pair}}
\newcommand{\forloop}{\text{\sf for-loop}}
\newcommand{\setS}{\mathcal{S}}
\newcommand{\joinpairs}{\text{\sf Join-Pairs}}
\newcommand{\postgres}{PostgreSQL}
\newcommand{\musicbrainz}{MusicBrainz}

\hypersetup{colorlinks, citecolor=black, filecolor=black, linkcolor=black,urlcolor=blue}

\newtheoremstyle{mytheoremstyle}{0pt}{0pt}{\itshape}{}{\bf}{.}{.5em}{} 

\theoremstyle{mytheoremstyle}
\newtheorem{lemma}[theorem]{Lemma}
\makeatletter
\xpatchcmd{\proof}{\topsep6\p@\@plus6\p@\relax}{}{}{}
\makeatother

\newcommand{\tocomplete}[1]{}

\newcommand{\reminder}[1]{}

\newcommand{\eat}[1]{}

\newcommand{\uniondp}{UnionDP}

\newcommand{\fullversion}[1]{}

\newcommand{\new}[1]{{{#1}}}
\newcommand{\revise}[1]{{{#1}}}

\usepackage{lipsum}

\newcommand\blfootnote[1]{%
  \begingroup
  \renewcommand\thefootnote{}\footnote{#1}%
  \addtocounter{footnote}{-1}%
  \endgroup
}

\newcommand{\todo}[1]{{\footnotesize \textcolor{red}{$\ll$\textsf{TODO #1}$\gg$}}}
\newcommand{\remove}[1]{}

{\itshape}{\rmfamily}

\setcopyright{none}
\renewcommand\footnotetextcopyrightpermission[1]{} 
\begin{document}
 \fancyhead{}
\title{Efficient Massively Parallel Join Optimization for Large Queries*}

\author{Riccardo Mancini}
\affiliation{%
  \institution{Scuola Superiore Sant'Anna}
  \country{}
}
\email{rickyman7@gmail.com}
\author{Srinivas Karthik}
\affiliation{
\institution{EPFL}\country{}
}
\email{skarthikv@gmail.com}
\author{Bikash Chandra}
\affiliation{
\institution{EPFL}\country{}
}
\email{bikash.chandra@epfl.ch}
\author{Vasilis Mageirakos}
\affiliation{
\institution{University of Patras}\country{}
}
\email{vasilis.mageirakos@gmail.com}
\author{Anastasia Ailamaki}
\affiliation{
\institution{EPFL \& RAW Labs SA}\country{}
}
\email{anastasia.ailamaki@epfl.ch}

\renewcommand{\shortauthors}{}

\begin{abstract}
\blfootnote{* This work was partially funded by the EU H2020 project SmartDataLake (825041)}
Modern data analytical workloads often need to run queries over a large number of tables. An optimal query plan for such queries is crucial for being able to run these queries within acceptable time bounds. However, with queries involving many tables, finding the optimal join order becomes a bottleneck in query optimization. Due to the exponential nature of join order optimization, optimizers resort to heuristic solutions after a threshold number of tables. Our objective is two fold: (a) reduce the optimization time for generating optimal plans;  and (b) improve the quality of the heuristic solution. 

In this paper, we propose a new massively parallel algorithm, {\mpdp}, that can efficiently prune the large search space (via a novel plan enumeration technique) while leveraging the massive parallelism offered by modern hardware (Eg: GPUs).  When evaluated on real-world benchmark queries with PostgreSQL, {\mpdp} is at least an order of magnitude faster compared to state-of-the-art techniques for large analytical queries. As a result,  we are able to increase the heuristic-fall-back limit from 12 relations to 25 relations with same time budget in PostgreSQL.  \new {Also,
in order to handle queries with even larger number of tables, we augment {\mpdp} to a well known heuristic, IDP$_2$ (iterative DP version 2) and a novel heuristic UnionDP}. By systematically exploring a much larger search space, these heuristics provides query plans that are up to 7 times cheaper as compared to the state-of-the-art techniques while being faster to compute.

\end{abstract}

%


\settopmatter{printfolios=true} 
\maketitle

\section{Introduction}
\label{sec:intro}

Data analytics in the modern world require processing of queries on large and complex datasets. In several business reporting tools, these analytical queries are automatically generated by the system. Such system generated queries tend to be very long (even up to megabytes in size). A single analytical query in such scenarios may contain up to several hundreds of tables, with even moderately sized queries having nearly 50 relations~\cite{dieu20091, chen2009partial, Neumann2018}.  Modern data analytical systems need to efficiently handle such large queries. The existence of these large query scenarios and inability of existing systems to handle them is described in \cite{dieu20091}.

\begin{figure}[t]
\begin{center}
\begin{boxedminipage}{3.25in}
{\setstretch{0}
\sf select o\_orderdate from lineitem, orders, part, customer \\
where {p\_partkey = l\_partkey} and {o\_orderkey = l\_orderkey} \\
and o\_custkey = c\_custkey
}
\end{boxedminipage}
\end{center}
\caption{Example TPC-H Query}
\label{fig:SPJ}
\end{figure}

In order to process such large analytical queries efficiently, finding optimal or near-optimal  query plans is essential. Finding an optimal join order for such queries with large number of relations is a challenging problem as the search space grows exponentially with the number of relations. 
For instance, {\postgres} takes as much as around 160 secs to find the optimal plan even for  a 21-relation join query\footnote{The optimization time was measured on a server with 2 Xeon CPUs on star join query}, while SparkSQL takes 1000 secs to plan an 18-relation \cite{krishnan2018learning}. Hence, current systems, resort to heuristics beyond a certain threshold number of relations (e.g. 12 relation in {\postgres}). 

Heuristics, however, may miss the optimal plan and, in such cases, the query execution time could be significantly higher than the optimal plan~\cite{Neumann2018}. Even though heuristics can produce sub-optimal plans, they are required to process queries with several 100s of relations. 

Our goal in this paper is to improve the performance of query optimizers for large join queries ($\geq$ 10 rels). Specifically, we aim to:
\begin{itemize}[noitemsep,topsep=0pt,leftmargin=*]
    \item Reduce the  query optimization time of \text{optimal} (or exact) algorithms. As a consequence, for a given  time budget,  increase the  heuristic-fall-back limit in terms of the number of relations.
    \item Improve the quality of heuristic techniques given a time budget. 
\end{itemize}
  
In this work, we focus on Dynamic Programming (DP) based join order optimization, which is typically used in current systems \cite{postGreSql, db2}. Moreover, we consider a solution without cross products similar to the one used in \cite{dpccp}, since it is well known that cross products do not form part of an optimal join order in most cases. 

The efficiency of any such DP  technique can be compared based on two key parameters:

\begin{enumerate}[noitemsep,topsep=0pt,leftmargin=*]
\item \emph{Number of join pairs evaluated:} 
DP algorithms typically follow an enumerate-and-evaluate approach. For the example query in Figure~\ref{fig:SPJ}, during plan exploration, it generates and evaluate if the following {\joinpairs} can form a valid (sub)plan: 1) (part, orders); 2) (part, lineitem); 3) (orders, lineitem). However, only {\joinpairs} (2) and (3)  are \emph{valid} as there is a corresponding join predicate in the query, while (1) is not valid since it has to be executed using a cross join. We provide a more precise description of valid {\joinpairs} in Section~\ref{sec:metric}. The fewer the invalid {\joinpairs} evaluated, the more efficient the algorithm is.

\item\emph{Parallelizability:} Another way to reduce the optimization time is to perform the join order optimization in parallel. \fullversion{For instance, finding the best (sub)plan at any level $i$ in {\dpsize} can be done in parallel}. 
For instance, $(part, lineitem)$ and $(orders, lineitem)$ {\joinpairs} can be evaluated in parallel. Note that \emph{not} all DP algorithms are easily parallelizable due to dependency between {\joinpairs} (detailed in Section~\ref{sec:metric}). The more parallelizable the algorithm is, the better is the performance. 

\end{enumerate}

A comparison of existing join order optimization techniques based on the above two parameters is shown in Figure~\ref{fig:landscape}. The Y-axis shows, for an input query, the  number of {\joinpairs} evaluated by different DP techniques normalized to the total number of valid {\joinpairs} for the query. 
The X-axis shows the parallelizability of the techniques. The evaluation is performed on a 20-relation query from the real world {\musicbrainz} dataset. 

\subsubsection*{\textbf{Optimal Solutions:} } 
The traditional {\dpsize} algorithm explores the search space in increasing sub-relation sizes. While  {\dpsub} enumerates all the powerset of relations in the subset \emph{precedence} order. Both {\dpsize} and {\dpsub} evaluate a lot of  invalid {\joinpairs} (around 500 times the valid {\joinpairs} as captured in the figure), and hence are inefficient. PDP~\cite{HanKLLM08} propose techniques to parallelize {\dpsize} but still evaluates a lot of invalid join pairs. 
Meister et al. in \cite{Meister2020}  leverage the GPU parallelism to further reduce query optimization time. They propose GPU parallel versions of {\dpsize} and {\dpsub} which scales better than the corresponding CPU parallel  ones. \new{Based on the enumeration style, we categorize {\dpsize} and {\dpsub} as \emph{vertex-based enumeration}}.

 \new{In contrast to \emph{vertex-based enumeration}, an \emph{edge-based enumeration}, {\dpccp}~\cite{dpccp}, evaluates only valid {\joinpairs}}. It enumerates the {\joinpairs} based on join graph dependencies which makes it difficult to parallelize.
Han et al.~\cite{HanL09}, parallelizes {\dpccp} but their producer-consumer paradigm for plan enumeration and costing 
limits its parallelizability~\cite{Meister2017}.

In this paper, we discuss a novel parallel join order optimization algorithm, {\mpdp} (Massively Parallel DP), which can be executed over GPUs (running with high degree of parallelism) or CPUs. The algorithm exploits the best of both {\dpsub} and {\dpccp} --  high parallelizability of {\dpsub} and minimum evaluation of {\joinpairs} from {\dpccp}. 
For the 20-rels example, the total {\joinpairs} evaluated by {\mpdp} is only twice that of valid {\joinpairs} for the query.

\subsubsection*{\textbf{Approximate/Heuristic Solutions:} } 
Since join order optimization is NP-Hard in general, for very large join queries, heuristics must be used. 
{\postgres} uses a genetic optimization based algorithm for such queries. When handling 100s of relations, an interesting approach, Iterative Dynamic Programming (IDP) \cite{Kossmann2000}, can be used which iterates over smaller join sizes and then combines them.
Recently, \cite{Neumann2018} proposes an adaptive optimization, {\lindp}, for handling large join queries by linearizing the DP search space.

 \begin{figure}
 	\centering
 	\includegraphics[width=0.43\textwidth,keepaspectratio=true]{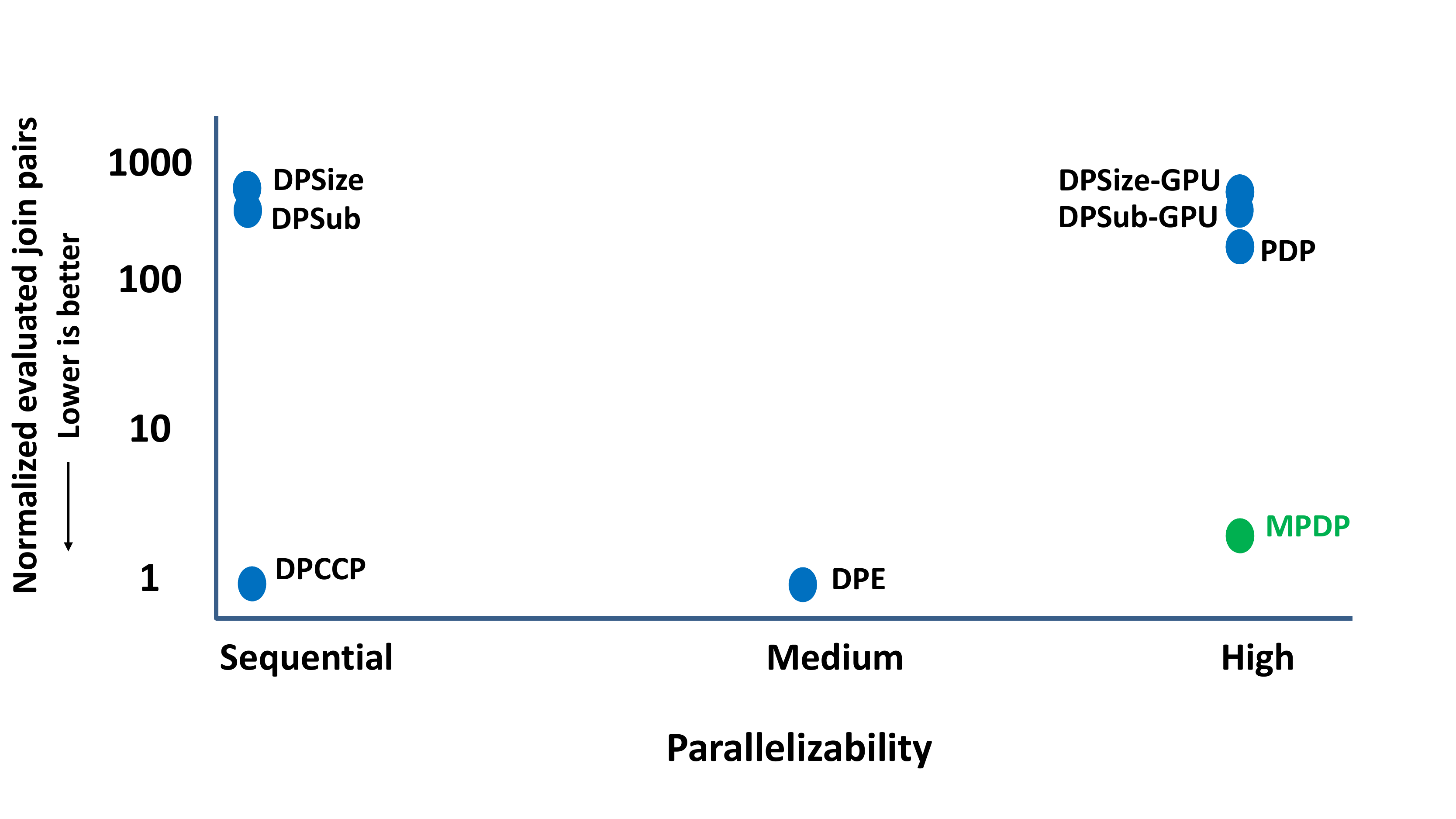}
 	\caption{Comparison of join ordering techniques}
 	\label{fig:landscape}
 \end{figure}
 
  We augment {\mpdp} with an existing heuristic algorithm, IDP. Due to the algorithmic efficiency and high parallelizability nature of {\mpdp}, we are able to systematically explore a much larger space compared to state-of-the-art solutions. 
  \new{We also develop a new heuristic technique, {\uniondp}, that leverages the join graph topology to get higher quality solution. The idea is to carefully partition the graph, use {\mpdp} on each partition, and systematically combine them. }

Our main contributions in this paper are as follows:
\begin{itemize}[noitemsep,topsep=0pt,leftmargin=*]
    \item We design, {\mpdp}, a new join order algorithm that is highly  parallelizable and evaluates only few invalid {\joinpairs}. We theoretically prove that the algorithm produces the optimal join order. Further, in case of commonly occurring tree join graphs, i.e. for star and snowflake join graphs, we prove that we do \emph{not} evaluate any invalid {\joinpairs}. 
    \new{We achieve this by proposing a novel plan enumeration technique that combines the vertex and edge-based enumeration. This hybrid enumeration is performed on carefully chosen subgraphs to make the algorithm massively parallelizable}.

    \item \new{In order to handle queries with even more relations than what is possible with optimal {\mpdp}, we propose two heuristic solutions that are algorithmically efficient and highly parallelizable. We discuss the heuristic solutions in Section~\ref{sec:approx}.} 

    \item We evaluate {\mpdp} (both exact and heuristic) on the open source {\postgres} database engine using queries on real world {\musicbrainz} dataset. To the best of our knowledge, our implementation on {\postgres},  is the \emph{first GPU-accelerated query optimizer} on a widely used  database system.  The implementation details are discussed in Section~\ref{sec:gpu_algo}.
    
    Our experimental results, in Section~\ref{sec:expt} show, for the exact solution, on the MusicBrainz dataset we get speed up of 80X compared to state-of-the-art parallel CPU algorithm ({\dpe}) on a 23 -relation query, and a factor 19X compared to state-of-the-art GPU based DP algorithm ({\dpsub}-GPU) on   a 26-relation query. Also, as a consequence, we can increase heuristic-fall-back limit from  12 to 25 relations in {\postgres} with same time budget.
    Both our heuristics can handle join queries with 1000 relations, and significantly improves over the state-of-the-technique in terms of quality of plans produced. Moreover, it also optimizes queries with 1000 relations under 1 minute. 
    

    \eat{This is in stark contrast to PostgreSQL which takes more than 5 mins (which times out for Z relations). Further we get speed up of X compared to state-of-the-art parallel CPU algorithm, and a factor Y compared to state-of-the-art GPU based DP algorithm. Also, as a consequence, we can increase heuristic threshold of 12-relation in Postgresql to 25-relations with the same optimization time limit (of 1 min).} 
    
\end{itemize}

We discuss relevant background in Section~\ref{sec:bg} and related work in Section~\ref{sec:relwork}. 

\eat{
\subsubsection{Organization:} The remainder of this paper is organized as follows: In Section~\ref{sec:bg}, problem framework and relevant background are enumerated. {\mpdp} for tree join graphs, its generalization and their theoretical analysis are  presented in Section~\ref{sec:gpu_algo}. This is followed by the {\mpdp}'s heuristic solution in Section~\ref{sec:approx}. Subsequently, GPU implementation details and related work are described in Section~\ref{sec:gpu_algo} and Section~\ref{sec:relwork}, respectively. Then, the experimental results are presented in Section~\ref{sec:expt}. Finally, conclusion and future work are enumerated in Section~\ref{sec:conc}. 
}

\section{Problem Framework and Background}
\label{sec:bg}

In this section, we discuss the problem framework and required notations.  Then we describe, {\dpsub} in detail -- the join order algorithm upon which we have built {\mpdp}. 
Finally, we also present key graph theory concepts that our solution uses.

\subsection{Valid {\joinpair} (\ccpPair)}
\label{sec:metric}
For a given query, we can represent the joins of the query as a graph $G(R,E)$, where the vertices $R=\{R_1, \cdots, R_n\}$ denote the set of all relations in the FROM clause of the query, while the edges, $E$, correspond to the inner join predicates in the query.
DP algorithms typically follow an enumerate-and-evaluate approach. \reminder{For  $S_{left},S_{right} \subset R$, it enumerates a  ${\text{\joinpair}}(S_{left},S_{right})$ and evaluates if it can form a valid sub-plan.} 
A {\joinpair}($S_{left},S_{right}$) is said to be valid (or can be joined to create a sub-plan) if all the following conditions hold true: 
\begin{enumerate}[noitemsep,topsep=0pt,leftmargin=*]
    \item Both $S_{left}$ and $S_{right}$ are non-empty subsets of $R$ 
    \item Induced subgraphs\footnote{Given a graph $G=(V,E)$ and a subset $S, \subset V$ of vertices, the induced subgraph of $S$ in $G$ is $G[S]=(S,E')$, where $E'=\{(a,b) | (a,b) \in E \land a \in S \land b \in S\}$ is the set of edges between nodes in $S$.} of both $S_{left}$ and $S_{right}$ in $G$ are connected
    \item $S_{left} \cap  S_{right} = \emptyset $  (disjoint) 
    \item $S_{left}$ is connected to $S_{right}$, i.e. there exists a vertex $v_l \in S_{left}$ and $v_r \in S_{right}$ such that there is an edge $(v_l,v_r) \in E$
\end{enumerate}

Note that any {\joinpair} $(S_{left}, S_{right})$ that satisfies all the above conditions is a Connected-subgraph Complement Pair ({\ccpPair}) as termed in \cite{dpccp}. We use the terms {\ccpPair} and  valid {\joinpair} interchangeably in the paper. 
{\ccpCount} represents the total number of {\ccpPairs} in a query, including the symmetric ones.
This count is dependent on join graph topology, and vastly varies between star, chain, cycle and clique graphs. 
However, for a given query, {\ccpCount} when profiled on any optimal DP  algorithm such as {\dpsize}, {\dpsub} and {\dpccp} will produce the same value.

Let us consider an example join graph with 8 relations shown in Figure~\ref{fig:tree}. Say that a DP algorithm enumerates a ${\text{\joinpair}}(S_{left}, S_{right})$ where $S_{left} = \{1,2,4\}$ and $S_{right} = \{6,7,8\}$. Since, there is no edge between these two sets in the join graph, it is \emph{not} a {\ccpPair}. While $S_{left} = \{1,2,4\}$ and $S_{right} = \{5,6\}$ is a {\ccpPair}, that can form a sub-plan with  \{$S_{left} \cup S_{right}$\} relations.

\begin{figure}[t]
\centering
\includegraphics[width=0.4\textwidth,keepaspectratio=true]{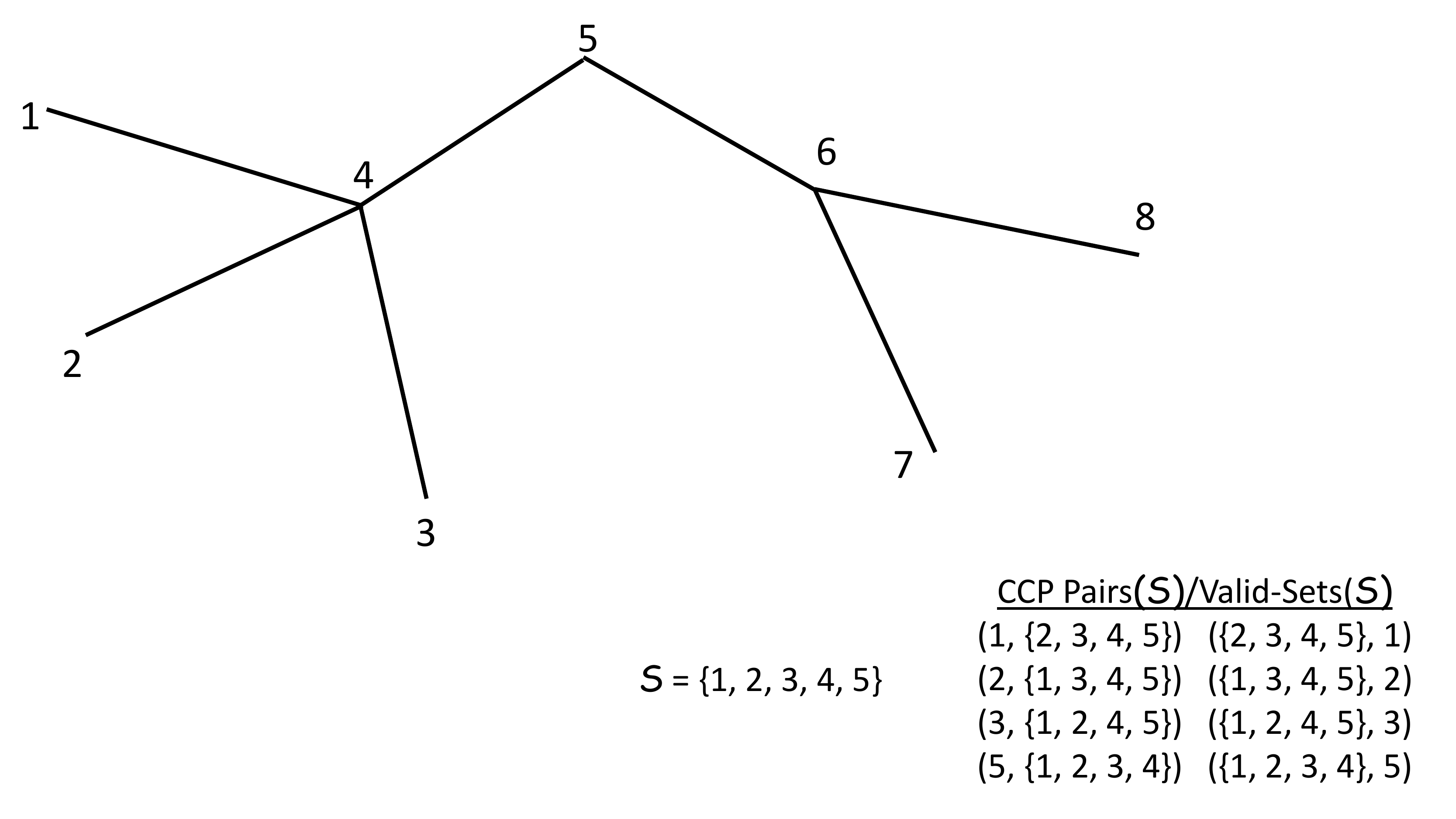}
\caption{Enumeration of {\ccpPairs} with Tree Join Graph}
\label{fig:tree}
\end{figure}

\subsubsection{Dependencies among {\joinpairs}}
\label{subsec:dependency}
Since we would like to develop a highly parallel algorithm, we should also keep into consideration the dependencies among {\joinpairs} that are enumerated. 
We say that a \joinpair($S_{left},S_{right}$) depends on \joinpair($S'_{left},S'_{right}$) if either $S_{left}$ or $S_{right}$ are the result of the join between $S'_{left}$ and $S'_{right}$.
In order to evaluate N {\joinpairs} in parallel, they must have no dependency among them. 
For instance, if the resulting joined-relation after joining a {\joinpair} has the same size $s$, then they will have no dependencies. Hence, {\dpsize} and the modified {\dpsub} in Algorithm~\ref{algo:generic_dpsub} can be  parallelized as discussed in~\cite{Meister2020,HanKLLM08}. 
It is also possible to define a dependency class based on the size of $S_{left}$~\cite{HanL09}.

\new{In this paper, we do not consider cross products which is a well accepted thumb rule in query optimization \cite{shanbhag2014optimizing,sqlserver_cross, dbjournal}. Furthermore, \cite{dbjournal, sqlserver_cross} suggest  to avoid cartesian joins  unless one has a reason such as joining small relations, or in case of division operator. This is primarily because cross products dramatically increase the search space but rarely produce better quality plans.
}

\subsubsection{Objective}
The objective is to develop an algorithm that, for any input query q, is able to find the optimal join order for $q$ without cross-products. This is achieved by the following sub-goals:  
(1) minimize the evaluation of {\joinpairs} that are  \emph{not}  {\ccpPairs}; 
(2) minimize dependency between {\joinpairs} while enumeration.

The first sub-goal is required  to minimize evaluating unnecessary {\joinpairs}, while the second is required for the algorithm to be highly parallelizable (i.e., scale to the massive parallelization offered by  GPUs).
Existing algorithms either fail to achieve the first or second objectives. 


\subsection{Generic DPSub Algorithm}
\begin{algorithm}
\setstretch{0.1}
	\caption{: Generic {\dpsub}}
	\label{algo:generic_dpsub}	
	\begin{algorithmic}[1]
		\REQUIRE $QI$: Query Information
		\ENSURE Best Plan
		\FORALL{$R_i \in QI.baseRelations$} 	
		   \STATE BestPlan($\{R_i\}$) = $R_i$ \label{alg1-line:initialization}
		\ENDFOR
		\FOR{$i:=2$ to $QI.querySize$} \label{alg1-line:Size}
		    \STATE $S_i = \{S\; | \; S \subseteq R \text{ and }|S| = i \text{ and S is connected}\}$ \label{alg1-line:Si}
            \FORALL {$\setS \in S_i$}
		    \label{alg1-line:SetS}
		    	\STATE //the following is done in parallel
		    \FORALL{$S_{left} \subseteq \setS$} \label{alg1-line:SubSetS}
		        \STATE {\innerCounter} ++ \label{alg1-line:innerCounter}
		        \STATE $S_{right} = \setS \setminus  S_{left}$ 
		        \STATE /*Begin CCP  Block */
		        \STATE \textbf{if} {$S_{right} == \emptyset$ or $S_{left} == \emptyset$}  \label{alg1-line:start-block}
		         \textbf{continue} 
		        \STATE \textbf{if} {not $S_{left}$ is connected } \label{alg1-line:connectedl} \textbf{continue}
		        \STATE \textbf{if} {not $S_{right}$ is connected } \label{alg1-line:connectedr} \textbf{continue}
		        \STATE \textbf{if} {not $S_{right} \cap S_{left} = \emptyset$}
		        \textbf{continue} \label{alg1-line:disjoint}
		         \STATE \textbf{if} {not $S_{right}$ is connected to $S_{left}$ }
		        \textbf{continue} \label{alg1-line:end-block}
		        \STATE /* End CCP Block */
		        \STATE {\ccpCount} $++$
		        \STATE CurrPlan = CreatePlan($S_{left},S_{right}$)
		        \IF {CurrPlan < $BestPlan(S)$} \label{alg1-line:bestplan1}
		        \STATE $BestPlan(\setS)$ = CurrPlan \label{alg1-line:bestplan2} 
		        \ENDIF
		    \ENDFOR
		    \ENDFOR 
		\ENDFOR
		\RETURN BestPlan(QI.baseRelations) //best plan for the query
	\end{algorithmic}
\end{algorithm}

We now present the generic {\dpsub} algorithm. 
The pseudo-code of {\dpsub} is shown in  Algorithm~\ref{algo:generic_dpsub}. For the sake of consistency, the presented pseudo-code is similar to the one used in~\cite{dpccp}.  

The algorithm iterates over all possible subset sizes $i$, and, for each size, it evaluates all non-empty connected subsets of
relations ${R_{1},..., R_{n}}$ (where $n$ is the number of relations in the query) of size $i$, constructing the best possible plan for each of them. 
The final plan is chosen at the root of the dynamic programming lattice.  Since the algorithm enumerates using subsets of relations, it is called Dynamic Programming Subset, or in short {\dpsub}. 

The algorithm starts with initializing $BestPlan(R_i)$ with its corresponding single relation $R_i$ (Line~\ref{alg1-line:initialization}). Here, $BestPlan(S)$ contains the best plan for any subset $S \subseteq R$ at any point of the DP algorithm. Then, the outermost nested  {\forloop} (Line~\ref{alg1-line:Size})
 collects all connected subsets, $S_i \subset R$, of  size $i$ in each iteration $i$ (Line~\ref{alg1-line:Si}). 
Further, in the middle nested {\forloop} (Line~\ref{alg1-line:SetS}), the goal  is to evaluate set
$\setS$ and get the best plan for it by the end of its iteration (Line~\ref{alg1-line:bestplan1} - Line~\ref{alg1-line:bestplan2}).  

Finally, in the innermost nested {\forloop}, all possible {\joinpairs}  are evaluated to see if it is  a {\ccpPair} based on the four conditions discussed in Section~\ref{sec:metric}. 
\eat{Specifically, 
for each  {\joinpair}$(S_{left}, S_{right})$  all the following {\ccpPair} conditions need to be satisfied: }
\begin{enumerate}
    \item Both $S_{left}$ and $S_{right}$ are non-empty subset of $\setS$ (Line~\ref{alg1-line:start-block}) 
    \item Both $S_{left}$ and $S_{right}$ should  be connected\footnote{For the sake of brevity, with ``subset $S$ is connected'', we mean that its induced subgraph $G[S]$ in the query graph $G$ is connected.} (Line~\ref{alg1-line:connectedl} - Line~\ref{alg1-line:connectedr}) 
    \item $S_{left}$ and $S_{right}$ should be disjoint (Line~\ref{alg1-line:disjoint})
    \item An edge should exist between $S_{left}$ and $S_{right}$ in its join graph (Line~\ref{alg1-line:end-block})
\end{enumerate}
All the above conditions are part of what we call as the \emph{Connected-subgraph Complement Pair} (CCP) block (Line~\ref{alg1-line:start-block} - Line~\ref{alg1-line:end-block}). 

If a {\joinpair}$(S_{left}, S_{right})$ happens to be a {\ccpPair}, then a plan, \textit{currPlan}, is created using the set $\{S_{left} \cup S_{right}\}$. If \textit{currPlan} is better than the current best plan for $\setS$, it is updated accordingly. \new{Based on the enumeration style, we refer {\dpsub} to as vertex-based enumeration. }

\begin{figure}[t]
\centering
\includegraphics[width=0.37\textwidth, height = 3.5cm]{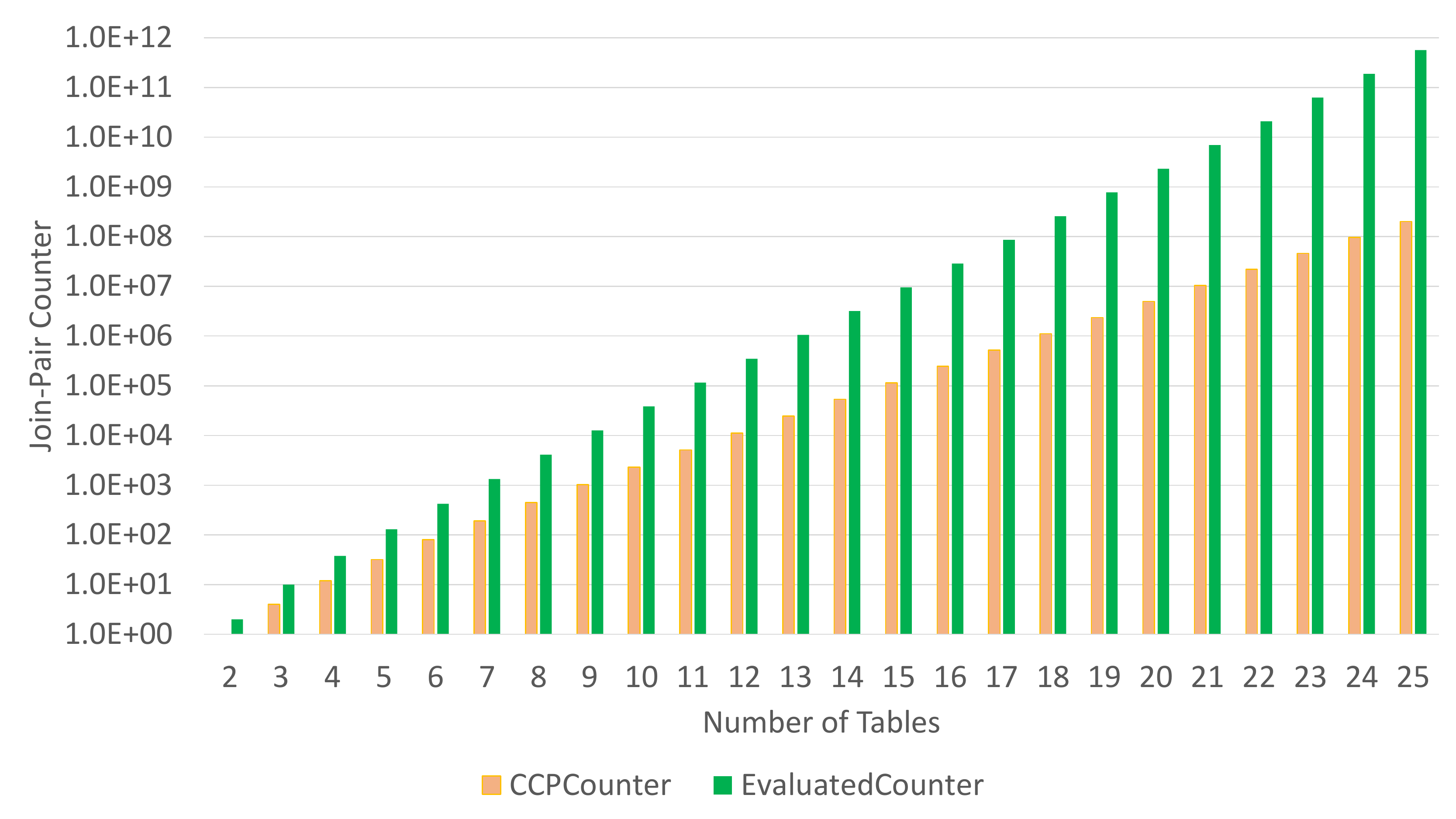}
\caption[{\genCount} vs {\ccpCount}]{{\innerCounter} and {\ccpCount} values for star join queries using {\dpsub}. {\innerCounter} can be around 2800 times larger than {\ccpCount}.}
\label{fig:gen-ccp-count}
\end{figure}

\new {
\subsubsection{Implementation Details: }
\label{sec:dpsub_impl}
All  sets  and  adjacency  lists  are  implemented  as  bitmap  sets. $S_i$ in  line~\ref{alg1-line:Si}  is  enumerated  using  the combinatorial system presented in \cite{dpccp}. While $S_{left}$ is obtained by enumerating from 1 to $2^{|S_i|} $, upon expanding the result of $S_i$ bits using parallel bit deposit (PDEP). Finally, checking the  connectivity of any set $S$ in the CCP block is done by using a \emph{grow} function from a random vertex in $S$ and checking if all vertices in $S$ are reachable  (\emph{grow} function is explained in more detail in Section~\ref{sec:grow}).}

\subsubsection{Parallelization:} 
 {\dpsub} is amenable for parallelization since sets $S_i$ are enumerated in increasing size as shown in Algorithm~\ref{algo:generic_dpsub}. Specifically, since enumeration of any ${\setS}$ of size $i$ is independent of each other, each iteration of the loop can be executed in parallel. Thus, the middle nested {\forloop} (Line~\ref{alg1-line:SetS}) can be parallelized. Further, even the computations in the innermost nested {\forloop}  (Line~\ref{alg1-line:SubSetS}) iterations are independent (and thus parallelizable), excluding the $BestPlan(S)$ update, which can be deferred to a later pruning step. 

\subsection{Shortcomings of {\dpsub}}
The main problem with {\dpsub} is that it evaluates the {\joinpairs} corresponding to the powerset of Set $\setS$ (Line~\ref{alg1-line:SubSetS} of Algorithm~\ref{algo:generic_dpsub}), and a small fraction of it ends up being {\ccpPair}s. In the algorithm, the total number of {\joinpairs} evaluated and the number of {\ccpPair}s are captured by {\innerCounter} and {\ccpCount}, respectively. 
In order to find the relative gap between the two counters,  we run {\dpsub} for star join graph queries, with varying number of relations. 
The results of this evaluation, as shown in Figure~\ref{fig:gen-ccp-count}, suggests that the gap between {\innerCounter} and {\ccpCount} increases with larger queries. Further, {\innerCounter} is around 2805 times larger (relatively) compared to {\ccpCount} at 25 relations.

Thus, although, {\dpsub} can be computed in a massively parallelizable manner, it evaluates  a lot of {\joinpairs} that are \emph{not} {\ccpPairs}. This is a motivation for us to design a parallel algorithm which minimizes the gap between {\innerCounter} and {\ccpCount}. 

\subsection{Relevant Graph Theoretic Terminologies}
\label{sec:graph_theory}
We now briefly discuss key graph theoretic terminologies that we use in our work. We use the graph shown in Figure~\ref{fig:ex-join-graph} as an example. 

\begin{enumerate}[noitemsep,topsep=0pt,leftmargin=*]
\item\emph{Cut Vertex}: 
 A cut vertex in an undirected graph is a vertex whose removal (and corresponding removal of all the edges incident on that vertex) increases the number of connected components in the graph. 
 For the example join graph in Figure~\ref{fig:ex-join-graph}, $\{4,5,9\}$ are cut vertices.

\item\emph{Nonseparable graph:} A graph $G$ is said to be separable if it is either disconnected or can be disconnected by removing one vertex. A graph that is not separable is said to be nonseparable.
 
\item\emph{Biconnected Component or block:} A biconnected component (or block) of a given graph is a maximal nonseparable subgraph. Note that a block contains some cut vertices of the graph, but does not have any cut vertex in the block itself. In our example, $\{1,2,3,4\}; \{4,5\}; \{5,9\}; \{6,7,8,9\}$ are blocks.

\item\emph{Block-Cut Tree:}
From a given graph G, we can build a bipartite tree, called block-cut tree, as follows. (1) Its vertices are the blocks and the cut vertices of G. (2) There exist an edge between a block and a cut vertex if that cut vertex is included in the block.
In our example, the block-cut tree would be a chain:
$\{1,2,3,4\} - 4 - \{4,5\} - 5 - \{5,9\} - 9 - \{6,7,8,9\}$.
\end{enumerate}

\section{{\mpdp}: A New Massively Parallel Optimal Algorithm}
\label{sec:algo}
In this section, we discuss our new Massively Parallel Dynamic Programming algorithm, {\mpdp}. For ease of presentation, we first discuss the simpler case when the join graph is a \emph{Tree}, and then generalize it to arbitrary join graphs. Commonly occurring star and snowflake join graphs belong to tree scenario \cite{dieu20091}. 

\fullversion{\subsection{{\dpsize} vs {\dpsub}} 
The most efficient CPU algorithm is {\dpccp}, which explores all and only the valid {\ccpPairs}. However, it does not lend itself to an easy and efficient parallelization. 
{\dpe}~\cite{HanL09} tries to overcome this problem using a producer-consumer paradigm: CCP pairs are enumerated in a producer thread and stored in a dependency-aware buffer, while consumers extract CCP pairs from the buffer and compute their join cost. 
However, Meister et al.~\cite{Meister2017} have found that this approach is effective only if the cost function is complex, due to the overheads introduced by managing the dependency buffer.

Another possibility would be {\dpsize}, which can be effectively parallelized on multiple CPU cores, as shown by the PDP algorithm~\cite{HanKLLM08}.
However, it enumerates orders of magnitude more {\joinpairs} than {\ccpPairs}, and this difference increases exponentially with the number of tables~\cite{dpccp}.
Overcoming this challenge in {\dpsize} is not trivial, in fact, being topology-agnostic, improvements based on the graph topology are not possible. 
}

\subsection{Tree Join Graphs}

\subsubsection{\textbf{Algorithm Description}}

The pseudo-code of {\mpdp} for the tree scenario is shown in Algorithm~\ref{algo:mpdp_trees}. To distinguish {\mpdp} from the general case, we refer to the algorithm as {\mpdptree}.  Note that the pseudo-code only contains the main {\forloop} corresponding to evaluating  set $\setS$. 
We have omitted the rest of the code since it is same as that used in {\dpsub} (Algorithm~\ref{algo:generic_dpsub}). Further, we have highlighted in \textcolor{red}{red} the difference in code between the two algorithms.  

The main idea of the algorithm is the following: Since the join graph is a tree, then the subgraph induced by $\setS$ is also a tree.  Then, the number of {\ccpPairs} of $S$ is exactly $i-1$, which corresponds to the {\joinpairs} formed by removing each edge in the tree induced by $\setS$ (Line~\ref{alg2-line:valid-join-pair}). In Figure~\ref{fig:tree}, for the tree graph, we also enumerate the {\joinpairs} created by removing edges for $\setS=\{1,2,3,4,5\}$. Given this insight, we only iterate over all {\ccpPairs} (Line~\ref{alg2-line:SubSetS}), create a plan for it and update the $BestPlan(\setS)$ accordingly. Thus, the algorithm do not incur any CCP conditions checking overheads. Resulting in {\innerCounter} being equal to {\ccpCount}.

Note that the idea of edge based {\joinpairs} join enumeration is very similar to the one used in {\dpccp}. However, the key difference is that {\dpccp} performs it at whole graph level, while we do it for subsets ${\setS}$ of size $i$. By doing this, apart from efficient enumeration, we maintain high \emph{parallelizability} of {\dpsub}, as both middle and innermost {\forloop} are parallelizable between their iterations. 

\begin{algorithm}
\setstretch{0.1}
	\caption{ {\mpdptree}}
	\label{algo:mpdp_trees}	
	\begin{algorithmic}[1]
		\FOR{$i:=2$ to $QI.querySize$} \label{alg2-line:Size}
		    \STATE $S_i = \{S\; | \; S \subseteq R \text{ and }|S| = i \text{ and S is connected} \}$ \label{alg2-line:Si}
		    \FORALL {$\setS \in S_i$} \label{alg2-line:SetS}
		    \STATE \textcolor{red}{\text{Valid-Join-Pairs($\setS$)} = Create {\joinpairs} by removing each edge in subgraph induced by  $\setS$} \label{alg2-line:valid-join-pair}
		    
		    \FORALL{\textcolor{red}{$(S_{left},S_{right}) \in \text{Valid-Join-Pairs}(\setS)$}} \label{alg2-line:SubSetS}
		        \STATE {\innerCounter} ++ \label{alg2-line:innerCounter}
		        \STATE {\ccpCount} $++$
		        \STATE CurrPlan = CreatePlan($S_{left},S_{right}$)
		        \IF {CurrPlan  < $BestPlan(\setS)$} \label{alg2-line:bestplan1}
		        \STATE $BestPlan(\setS)$ = CurrPlan \label{alg2-line:bestplan2} 
		        \ENDIF
		    \ENDFOR
		    \ENDFOR 
		\ENDFOR
	\end{algorithmic}
\end{algorithm}

{
\subsubsection{\textbf{Proof of Correctness}}
In order to show that our proposed algorithm is correct, we need to prove the following:
\begin{enumerate}
    \item Only the {\ccpPairs} corresponding to connected set $\setS$ are enumerated, i.e. any \joinpair $(S_{left}, S_{right}) \in Valid-Join-Pairs(S)$ is a {\ccpPair}
 (Lemma~\ref{lemma:tree:only}). 
    \item All the {\ccpPairs} corresponding to connected set $\setS$ are enumerated (Lemma~\ref{lemma:tree:all}).
    \fullversion{\item Every {\ccpPair} corresponding to connected set $\setS$ is enumerated only once (Lemma~\ref{lemma:tree:once}).}
\end{enumerate}{}

\begin{lemma}\label{lemma:tree:only}
Any \joinpair $(S_{left}, S_{right}) \in \text{Valid-Join-Pairs}(\setS)$ is a {\ccpPair} (Line~\ref{alg2-line:valid-join-pair}).
\end{lemma}
\begin{proof}
Trivially, both  $S_{left}, S_{right} \neq \emptyset$. Similarly, $S_{left} \cup S_{right} = \setS$ and $S_{left} \cap S_{right} = \emptyset$ also hold true since the {\joinpair} is created by removing a single edge in the tree. 
Then, further both $S_{left}, S_{right}$ are connected, if not, then $\setS$ would also be disconnected which leads to the contradiction that $\setS$ is connected. 
\fullversion{The final condition that there is an edge between $S_{left}$ and $S_{right}$ is also satisfied, which is exactly the  edge that was removed to form the {\joinpair}.}
\end{proof}

\begin{lemma}\label{lemma:tree:all}
{\dpsub} and {\mpdptree} evaluate same set of    {\ccpPairs}.
\end{lemma}
\begin{proof}
Consider a {\ccpPair} $(S_{left}, S_{right})$ evaluated by {\dpsub}. 
Then, exactly the same {\joinpair} would also be enumerated by removing the edge that exist, by definition of {\ccpPair}, between the two sets. 
\end{proof}

\fullversion{
\begin{lemma}\label{lemma:tree:once}
All the {\ccpPairs} corresponding to connected set $\setS$ are enumerated only once.
\end{lemma}
\begin{proof}
\new{ The proof follows easily by contradiction, and thus skip it in the interest of space.} 
 this by contradiction: Say that a {\joinpair} $(S_{left},S_{right})$ is enumerated twice. Then there exist two edges whose individual removal  would result in {\ccpPair} $(S_{left},S_{right})$. By definition of {\ccpPair}, this results in $\{S_{left} \cup S_{right}\}$ having a cycle. Hence a contradiction. 
\end{proof}}

From the above lemmas, the following theorem can be inferred:
\begin{theorem}
{\mpdptree} finds the optimal join order while evaluating only {\ccpPairs} (meeting the {\ccpCount} lower bound)
\end{theorem}
}
\subsection{Generalization}
\label{sec:mpdp_gen}
\new{
After having seen {\mpdptree} for tree join graphs, generalizing to join graph with cycles would pose the  following challenges:
\begin{enumerate}[noitemsep,topsep=0pt,leftmargin=*]
    \item \emph{Edge-based enumeration: } With cyclic graphs, removing edges as in tree scenario may not form a join-pair. For instance in Figure~\ref{fig:ex-join-graph}, removing edge (1,4) would not form a {\joinpair}. 
    \item \emph{Vertex-based enumeration: } Boils down to the conventional {\dpsub} which falls prey to highly  inefficient enumeration. 
\end{enumerate}
\emph{Our contribution is a novel enumeration technique which is a hybrid of
vertex and edge-based enumeration that results in: (1) efficient enumeration (i.e. close to minimum  {\joinpair} evaluation); (ii) highly parallelizable.}  This is achieved by identifying blocks (or biconnected components) in the graph. Then, we perform: a) edge-based enumerated along the cut edges between blocks; b) vertex-based enumeration within the blocks. Further, a vertex-based enumeration within a block happens by creating {\joinpairs} within each block. 
Then, using the edge-based enumeration along cut-edges, we create a {\joinpair} for the set $\setS$ using the block {\joinpair} as the seed nodes. 
We show the correctness of the algorithm by mapping the {\joinpair} at the block-level to the {\joinpair} at the set ${\setS}$ level. Since the expensive vertex-based enumeration is just limited to blocks, the number of join-pair evaluation reduces from $2^{|\setS|}$ to $\mathcal{O}(\text{no. of blocks} * 2^{\text{ max. block size}})$. For our cyclic graph example, it reduces from 512 to just 32.}

\subsubsection{\textbf{Grow Function}}
\label{sec:grow}
\new{The \emph{grow} function takes as input a set of \emph{source} nodes and \emph{restricted} nodes (superset of \emph{source} nodes), and output all the nodes in the \emph{restricted} set that are reachable from source nodes.} This is achieved by iteratively adding all the restricted nodes such that they are connected to at least one node in the \emph{source} set, and \emph{growing} the source set by adding to it. For example in Figure~\ref{fig:ex-join-graph}, if \emph{source} nodes are  \{1,2,3\} and restricted nodes are \{1,2,3,4,5,9\}, then grow function returns \{1,2,3,4,5,9\}.

\begin{figure}
\centering
\includegraphics[width=0.43\textwidth]{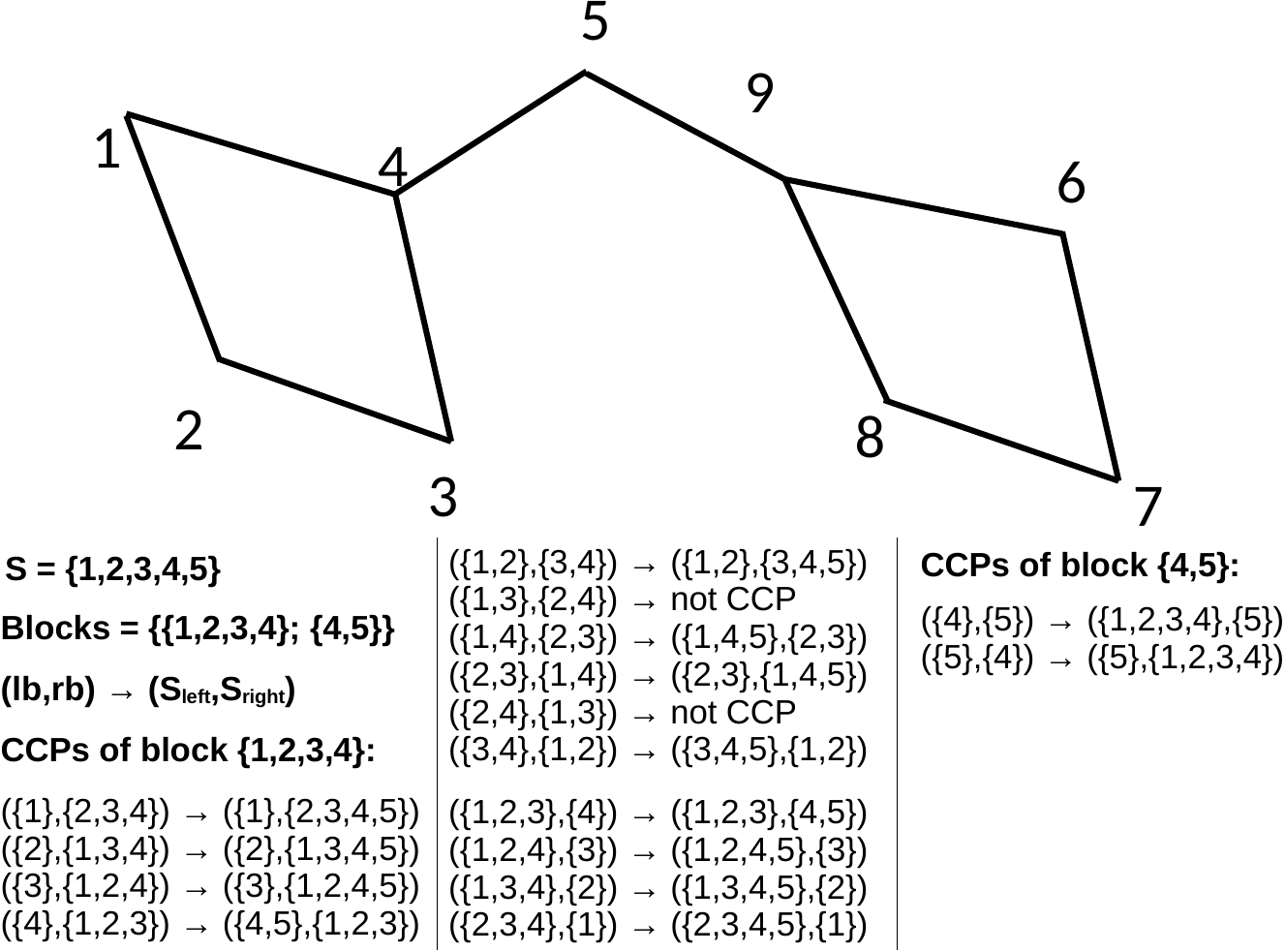}
\caption[Example Join Graph with 9 Relations]{Example Join Graph with 9 Relations}
\label{fig:ex-join-graph}
\end{figure}

\subsubsection{\textbf{Algorithm Description}}
\label{sec:alg_des_gen}
\begin{algorithm}
\setstretch{0}
	\caption{: {\mpdp} generalization (with cycles)}
	\label{algo:mpdp_generalization}	
	\renewcommand{\algorithmicrequire}{\textbf{Inputs:}}
	\renewcommand{\algorithmicensure}{\textbf{Output:}}	
	\begin{algorithmic}[1]
	\FOR{$i:=2$ to $QI.querySize$} \label{alg3-line:Size}
		\STATE $S_i = \{S\; | \; S \subseteq R \text{ and }|S| = i \text{ and S is connected}\}$ \label{alg3-line:Si}

		\FORALL {$\setS \in S_i$} \label{alg3-line:SetS}
		\STATE $BLOCKS$ $\leftarrow$ $\text{Find-Blocks}(\setS, QI)$ \label{alg3-line:findblocks}
		\FORALL{$block$ $\in$ $BLOCKS$} \label{alg3-line:iterate_blocks}
    		\FORALL{$lb$ $\subset$ $block$, $lb \neq \emptyset$} \label{alg3-line:iterate-block}
    		    \STATE {\innerCounter}++
    		    \STATE $rb$ $\leftarrow$ $block \setminus lb$ \label{alg3-line:rb}
    		    
    		    \STATE /*Begin CCP  Block */
		        \STATE \textbf{if} {$rb == \emptyset$ or $lb == \emptyset$} \label{alg3-line:start-block}
		         \textbf{continue} 
		        \STATE \textbf{if} {not $lb$ is connected } \label{alg3-line:connectedl} \textbf{continue}
		        \STATE \textbf{if} {not $rb$ is connected } \label{alg3-line:connectedr} \textbf{continue}
		        \STATE \textbf{if} {not $rb \cap lb == \emptyset$}
		        \textbf{continue} \label{alg3-line:disjoint}
		         \STATE \textbf{if} {not $rb$ is connected to $lb$ }
		        \textbf{continue} \label{alg3-line:end-block}
		        \STATE /* End CCP Block */
    		    \STATE {\ccpCount} $++$
    		    
    		    \STATE $\setS_{left}$ $\leftarrow$ $grow(lb, \setS \setminus rb)$ \label{alg3-line:grow_left}
    		
    		\STATE \new{$\setS_{right} $ $\leftarrow$  $\setS \setminus \setS_{left}$}
    		\label{alg3-line:grow_right}

		        \STATE CurrPlan = CreatePlan($S_{left},S_{right}$)
		        \IF {CurrPlan < $BestPlan(\setS)$} \label{alg3-line:bestplan1}
		        \STATE $BestPlan(\setS)$ = CurrPlan \label{alg3-line:bestplan2} 
    		    \ENDIF
    		\ENDFOR
		\ENDFOR
		\ENDFOR
		\ENDFOR
	\end{algorithmic}
\end{algorithm}

We now present the generalized {\mpdp} algorithm, the pseudocode of which is presented in Algorithm~\ref{algo:mpdp_generalization}. 
For the outermost nested {\forloop} (Line~\ref{alg3-line:SetS}), the key difference from {\dpsub} is that, instead of iterating over all subsets of ${\setS}$,  we iterate over all subsets of each block in ${\setS}$. This block-level enumeration results in significantly lower {\joinpair} evaluation (cf. Section~\ref{sec:mpdp:analysis}). 
\fullversion{In particular, it is equivalent to {\ccpCount} if all blocks of $\setS$ are fully connected (cf. Lemma~\ref{lemma:generic:opt}), e.g. when the subgraph induced by $\setS$ is a tree or a clique. It is implemented in the GPU  using a warp-level parallelism~\cite{bfc-bicc}. 

}

\new{We first identify all the blocks\footnote{\new{In the scenario where we do not find blocks, it boils down to the case of pure vertex-based enumeration, i.e. {\dpsub}.}} in $\setS$ using Find-Blocks function (Line~\ref{alg3-line:findblocks}). The Find-Blocks function can be implemented using the DFS-based Hopcroft and Tarjan algorithm~\cite{hopcroft-bicc} -- a parallel version of it also exist~\cite{bfs-bicc}. \footnote{For intuition, the cyclic graph can be represent as a Block-cut tree, i.e. a tree of blocks connected by cut-edges. Note that creation of block-cut tree is not necessary to find the blocks.}}

Next, for each block (Line~\ref{alg3-line:iterate_blocks}), we iterate over all subsets, $lb$ of the block, compute its complement within the block, $rb$, and check that they form a {\ccpPair} for the block (Lines~\ref{alg3-line:start-block}~-~\ref{alg3-line:end-block}).
Next the key step is to create a {\ccpPair} wrt $\setS$ using the {\ccpPair} $(lb, rb)$. 
This is achieved by using the \emph{grow} function  on 
computing $\setS_{left}$, the set of reachable nodes within the restriction set, $\setS-rb$, from the source nodes in $lb$ (Line~\ref{alg3-line:grow_left}). 
Likewise, ${\setS_{right}}$, reachable nodes with the restriction set, $\setS-lb$, that can be visited starting from the source nodes in $rb$ (Line~\ref{alg3-line:grow_right}). Finally, for each {\ccpPair} $(\setS_{left},\setS_{right})$, a plan is created and $BestPlan(\setS)$ is updated accordingly.  

\emph{Parallelizability:} By processing $\setS$ over blocks,  parallelizability of the algorithm is not impacted compared to {\dpsub}. This is because, all the three nested {\forloop}s (Line~\ref{alg3-line:SetS}, Line~\ref{alg3-line:iterate_blocks}, Line~\ref{alg3-line:iterate-block}), including newly added innermost nested {\forloop} which can be run in parallel.  

\subsubsection{\textbf{Proof of Correctness}}
The proof structure is along the lines of Tree scenario. We show the following with respect to set $\setS$: 1) All the {\ccpPairs}  are enumerated  (Lemma~\ref{lemma:generic:all})
; 2) All pairs $(S_{left},S_{right})$ are {\ccpPairs}  (Lemma~\ref{lemma:generic:only}); \fullversion{3) Every {\ccpPair}  is enumerated only once (Lemma~\ref{lemma:generic:once})}.

\begin{lemma} \label{lemma:generic:all}
{\dpsub} and {\mpdp} enumerates the same set of {\ccpPair}s.
\end{lemma}
\begin{proof}
Let's take any {\ccpPair}$(\setS_{left},\setS_{right})$ enumerated by {\dpsub}. Since $\setS_{left}$ is connected to $\setS_{right}$, there exists at least one edge connecting a node in $\setS_{left}$ and a node in $\setS_{right}$. We want to prove that these edges are all inside the same block. If there is only one edge, it is obvious, since one edge can belong to only one block.

By contradiction, let's assume there are 2 edges, $(n_l,n_r), (n'_l,n'_r) \mid n_l,n'_l \in \setS_{left} \land n_r,n'_r \in \setS_{right}$, which are contained in different blocks. 
Since $\setS_{left}$ and $\setS_{right}$ are connected, this would imply the existence of a cyclic path passing through these two blocks. This contradicts the maximality property of the block.

Since these edges are all inside the same block, then we can identify $lb = \setS_{left} \cap block$ and $rb = \setS_{right} \cap block$. We want to prove that $(lb,rb)$ is a {\ccpPair} for the $block$. 
In fact, by construction: $lb,rb \neq \emptyset \land lb,rb \in \setS$; $lb \cap rb = \emptyset$; $lb$ connected to $rb$.
We only need to prove that $lb$ and $rb$ induce connected subgraphs.
Let's consider $lb$. By contradiction, let's assume that it is not connected and let's take $n_l,n'_l \in lb$, belonging to different connected components of the subgraph induced by $lb$. Since $\setS_{left}$ is connected, then there is a path outside the block that joins $n_l$ to $n'_l$. Since $n_l,n'_l \in block$, there exists a path within the block connecting these two nodes. This implies the existence of a cyclic path spanning multiple blocks, which contradicts the property of maximality of the block.
The same is true for $rb$.

Since $(lb,rb)$ is a {\ccpPair} for the $block$, it will be enumerated by {\mpdp}, because it exhaustively enumerates all {\joinpairs} in the block. Finally, we need to show that $\setS_{left} = grow(lb,block \setminus rb)$ and that $\setS_{right} = grow(rb,block \setminus lb)$.

Considering $lb$, since the edges connecting $\setS_{left}$ to $\setS_{right}$ are only inside the block containing $lb$, and since \emph{grow} is restricted to $\setS \setminus rb$, $S_{left}$ cannot contain any node in $S_{right}$. Furthermore, since $S_{left}$ is connected, \emph{grow} will visit all nodes in $S_{left}$. Therefore, since there are no nodes in $\setS$ not in $\setS_{left}$ and $\setS_{right}$, $\setS_{left} = grow(lb,block \setminus rb)$. 
Likewise, the same can be demonstrated in the same way also for $rb$. Therefore, {\ccpPair}$(\setS_{left},\setS_{right})$ will also be enumerated by {\mpdp}.
\end{proof}

\begin{lemma} \label{lemma:generic:only}
Any \joinpair$(S_{left}, S_{right})$ constructed from the {\ccpPair} $(lb,rb)$ for the block is a {\ccpPair} for $\setS$.
\end{lemma}
\begin{proof}
Trivially, both  $S_{left}, S_{right} \neq \emptyset$. 
Further, $S_{left} \cap S_{right} = \emptyset$ also hold true. 
By contradiction, let's assume that $S_{left} \cap S_{right} \neq \emptyset$. Therefore, there exists a node $n \mid n \in S_{left} \cap S_{right} \land n \not\in block$ that is reachable from a node $n_l$ in $lb$ and a node $n_r$ in $rb$, creating a cycle. This implies that there exists a path outside the block joining $lb$ and $rb$, which contradicts the maximal property of the block. 

In addition, both $S_{left}, S_{right}$ induce connected subgraphs, as they are the result of the \emph{grow} function on connected subsets. Finally, $\setS_{left}$ is also connected to $\setS_{right}$, since $lb$, subset of $\setS_{left}$, is connected to $rb$, subset of $\setS_{right}$.
\end{proof}

From the above lemmas, the following theorem can be inferred:
\begin{theorem}
{\mpdp} finds the optimal join order
\end{theorem}

\subsubsection{\textbf{Analysis}} \label{sec:mpdp:analysis}
We now analyze the number of {\joinpairs} evaluated by {\mpdp} in comparison to {\dpsub}. 
\fullversion{
\begin{enumerate}
    \item {\mpdp} is more efficient than {\dpsub} (lemma~\ref{lemma:generic:complexity});
    \item {\mpdp} is optimal on graphs whose blocks are all fully connected (lemma~\ref{lemma:generic:opt}), e.g. in trees\footnote{Note that in a tree, all vertices are cut-vertices and all edges $(a,b)$ form a block $\{a,b\}$, which is obviously fully-connected.} and cliques\footnote{Also note that in a clique, there is only one block which is fully-connected.}.
\end{enumerate}
}

\begin{lemma} \label{lemma:generic:complexity}
For a given set $\setS$, the number of subsets evaluated by {\mpdp} is lower than {\dpsub}.
\end{lemma}
\begin{proof}
Let's start by observing that all subsets in a block are enumerated, which are $2^b$, where $b$ is the size of the block. 
This implies that, for the given set $\setS$, the total number of evaluated subsets is $\sum_{blocks} 2^{b-1}$.
Furthermore, we also have that $1+\sum_{blocks} b-1 = n$, where $n = |\setS|$.
Therefore, $\sum_{blocks} 2^{b-1} \leq 2^{n-1}$, which can be rewritten as $\sum_{blocks} 2^b \leq 2^n$, where $2^n$ is the number of subsets evaluated by {\dpsub}. Finally,  the time complexity for each set $\setS$ from $O(2^n)$ to $O(B*2^b)$, where $n = |\setS|$, $b$ is the max block size ($b \leq n$), and $B$ is the number of blocks in subgraph $\setS$.
\end{proof}

\begin{lemma} \label{lemma:generic:once}
All the {\ccpPairs} corresponding to the connected set $\setS$ are enumerated only once.
\end{lemma}
\begin{proof}
Since we've proven that each {\ccpPair} has edges connecting the two parts in only one block, it is impossible that the same {\ccpPair} is enumerated starting from different blocks. 
Furthermore, two different {\ccpPairs}, $(lb,rb)$ and $(lb',rb')$, within the same block will produce different {\ccpPairs} in $\setS$, say \linebreak  $(S_{left}, S_{right})$ and $(S'_{left}, S'_{right})$. 
By contradiction, if $S_{left} = S'_{left} \land S_{right} = S'_{right}$, then also $lb = S_{left} \cap block = S'_{left} \cap block = lb' \land rb = S_{right} \cap block = S'_{right} \cap block = rb'$, which is a contradiction.
\end{proof}

\begin{lemma} \label{lemma:generic:opt}
In {\mpdp}, if all blocks in the graph are fully connected, then {\innerCounter} = {\ccpCount}.
\end{lemma}
\begin{proof}
If a block is fully connected, then all pairs of disjoint subsets $(lb,rb) \in block$, such that $lb \cup rb = block$, are {\ccpPairs}. 
In fact, the full connectivity of the block implies that also $lb$ and $rb$ are connected, and that $lb$ is connected to $rb$.
Therefore, all enumerated {\joinpairs} are also {\ccpPairs}.
\end{proof}

\fullversion{
\reminder{ May be retained only if needed}

\subsection{Parallelization}

As previously discussed, parallelizing a DP algorithm that enumerates join pairs by levels of increasing size is trivial. In fact, all sets $\setS$ of the same size can be processed in parallel without dependencies. 
Therefore, we can simply execute the for-loop at Line~\ref{alg3-line:SetS} in parallel.

In our multi-threaded CPU implementation, there is a shared atomic counter which starts at 0 and represents the index of the next set $\setS$ to unrank. Worker threads atomically fetch and increment this counter by a fixed value $N$, which represents the size of the work item, storing it in a local counter $i$. Then, for each index from $i$ to $i+N-1$, they unrank the corresponding set using the combinatorial unranking algorithm presented by Meister et al.~\cite{Meister2017}, find its best {\joinpair}, and save the resulting plan in a local list. When a thread finishes its work for this level, it stores all plans from its list to the memo table, protected by a mutex. 
The memo table has one hash-table for each size, so that inserting elements in one of the hash-tables does not invalidate references to join relations of lower size.

In our experiments, the overhead of synchronization through the atomic counter and the mutex were negligible compared to the set evaluation phase. Other approaches, like statically splitting the search space, could yield unbalanced work items. 
}

\section{Heuristic Solutions}
\label{sec:approx}

Although {\mpdp} is efficient, parallelizable and runs well on GPUs, join order optimization is an NP-Hard problem.  The time taken for optimization increases exponentially. Hence, for very large joins, we need to apply a heuristic optimization technique. \eat{We use the well known Iterative Dynamic Programming (IDP)~\cite{Kossmann2000} technique and integrate {\mpdp} with it. \todo{rewrite}} \new{We propose two heuristic solutions: 1) augmenting {\mpdp} into an existing heuristic technique, IDP; 2) a novel join-graph conscious heuristic, {\uniondp}.}

\subsection{Iterative Dynamic Programming}
\label{sec:idp2}

\fullversion{
\begin{algorithm}
	\caption{: IDP$_2$}
	\label{algo:idp2_goo}	
	\renewcommand{\algorithmicrequire}{\textbf{Inputs:}}
	\renewcommand{\algorithmicensure}{\textbf{Output:}}	
	\begin{algorithmic}[1]
		\REQUIRE $G = (V, E)$: query graph\\
		         $k$: maximum size of re-optimized subtree\\
		         $dp$: optimal join order DP algorithm\\
		         $heu$: fast heuristic algorithm\\
		\STATE $T \leftarrow heu(V, G)$ // join tree found by heuristic
		    \label{line:idp2:heuristic}
		\WHILE{$|T| > 1$} \label{line:second_part_start}
		    \STATE pick $T'$ subtree of $T$ such that: $1 < |T'| \leq k$ and $C(T')$ is maximal
		    \STATE // optimize join order of relations in $T'$ using $dp$
		    \STATE replace $T'$ in $T$ with $dp(T', G)$ as a temporary table
		    \STATE // optionally, stop whenever a budget is exhausted (e.g. time)
		\ENDWHILE \label{line:second_part_end}
		\STATE expand all temporary tables in $T$ back to their join tree
		\RETURN $T$
	\end{algorithmic}
\end{algorithm}
}

\new{Kossmann et al.~\cite{Kossmann2000} proposed two versions of Iterative Dynamic Programming (IDP) techniques. The first one, IDP$_1$, initially  builds join order plans up to a given number of relations $k$ using the exhaustive algorithm, picks the lowest cost plan for $k$ relation joins, materializes it and then uses it as a single relation for subsequent iterations.
However, IDP$_1$ has a time complexity of $O(n^k)$, making it viable only for small values of $k$ and $n$. Optimizing large queries requires a too small value of $k$, that negatively impacts the quality of plans.}

The second version of IDP, IDP$_2$, applies the heuristic a priori, first generating a tentative plan and then optimizing it.
It is made up by two components: 
(1) \emph{Initial Join Order: } A heuristic algorithm to build an initial join plan (or join order). 
(2) \emph{Iterative DP: } The constructed join tree in the above step  is the input to this component. The idea is to use an optimal join order DP algorithm to optimize the most costly parts of the join plan. At each step, the most costly subtree up to size $k$ is selected for optimization, $dp$ is run on its relations to find the optimal plan $T'$, and finally it is replaced in $T$ by a single temporary table, therefore reducing the size of $T$ by $|T'|-1 \in [1, k-1]$. 
The loop will run until only one temporary table representing the whole query remains in $T$. 
At this point, all temporary tables are reverted to their optimal tree form before returning $T$.
Note that it is also possible to stop the while loop at any iteration and obtain an acceptable plan, based on a given time budget. Its time complexity is $O(n^3)$, if $n >> k$.

\subsubsection{IDP$_2$ with MPDP}
\label{sec:mpdp-idp}
\new{
MPDP can be incorporated into IDP$_2$ by replacing the $dp$ algorithm  with {\mpdp}. We use IDP$_2$, since it performs better than IDP$_1$ for very large join queries. 
\fullversion{Applying {\mpdp} to $IDP_2$ does not require any change to the structure of $IDP_2$.} 
Also the advantage of using {\mpdp} inside $IDP_2$, instead of another exact DP algorithm on CPU, is that it 
allows for a bigger $k$ for the same planning time. 
\footnote{In our experiments on snowflake schema, we were able to use $k$ value of up to 25 for GPU accelerated {\mpdp}, with the optimization time at 100 tables being 550ms.} 
This is beneficial since the algorithm explores a much larger search space and, therefore, it may be able to find a better plan.
{\mpdp} is called from within $IDP_2$ with the correct subset of the query information that is to be processed at this state. }
\subsection{\uniondp}
\label{sec:uniondp}
\new{ The main issue with IDP$_2$ is that, due to the initial plan choice and its greedy nature, it might get stuck on a poor local optima, resulting in suboptimal plan choices. Hence, we design a novel heuristic, {\uniondp}, that for the \emph{first-time} leverages the graph topology for such large queries. 

The key idea of {\uniondp} is to partition the graph into tractable sub-problems, solve each of the sub-problems with {\mpdp} optimally, then, recursively build the solution to  the original problem from these sub-problems. In order to produce quality plans in reasonable times, the challenge would be to satisfy the following requirements:
\begin{enumerate}[noitemsep,topsep=0pt,leftmargin=*]
    \item \emph{Partition Size:} The size of each partition should be less than a threshold value such that the partition can be optimized efficiently  by {\mpdp}. Note that all the partition sizes ideally should be close to the threshold value. If the partition sizes are too less than the threshold, then, this possibly increases the optimization time and may results in lesser quality plans as partitions inhibits search space exploration.  
    \item \emph{Weight of Cut Edges:} We assign the weight of edges to be cost (using a cost model) of joining the relations across the edge. The sum of weight of cut edges of the partitions needs to be as high as possible. 
    This is because more costly join needs to be as late as possible in the plan tree following the convention that higher selectivity predicates are applied earlier in the plan tree.  This requirement typically trades-off with (1). 
\end{enumerate}

}

\begin{algorithm}
\setstretch{0}
    \new{
	\caption{\uniondp}
	\label{algo:uniondp}	
	\renewcommand{\algorithmicrequire}{\textbf{Inputs:}}
	\renewcommand{\algorithmicensure}{\textbf{Output:}}	
	\begin{algorithmic}[1]
		\REQUIRE $G = (R, E)$: query graph\\
		         $k$: maximum number of relations in a partition\\
		         
        \IF {$nRels(G) \leq k$} \label{alg4-line:end_condition_start}
        \STATE $T \leftarrow \textit{\mpdp}(G) $ // optimal sub-plan found by {\mpdp} 
        \RETURN $T$
        \ENDIF \label{alg4-line:end_condition_end}
        \STATE For any edge, assign leftRelSet and rightRelSet to be set of rels across the edge
        \STATE assignEdgeWeights() //Assign edge weights using the cost model 
        \label{alg4-line:assign_weights}
        \STATE \textit{makeSet}(G) // create a disjoint set for each relation in graph
        \label{alg4-line:makeset}
        \FORALL {edges in increasing order of $size$(leftRelSet + rightRelSet) }
        \label{alg4-line:edges_itr}
        \STATE //Use weights of edges in case of tie for the above
        
        \IF {$size$(leftRelSet) + $size$(rightRelSet) $\leq k$} 
        \label{alg4-line:union_condition}
        \STATE \textit{Union}(leftRelSet, rightRelSet) \label{alg4-line:union}
        \ENDIF
        \ENDFOR
        \STATE /*End of Partition Phase */
        \FORALL {induced subgraphs of the disjoint sets}
        \STATE $T' \leftarrow {\mpdp}(subgraph) $ 
        \label{alg4-line:sub_optimize}
        \STATE createCompositeNode($T'$)
        \label{alg4-line:create-composite}
        \ENDFOR
        
        \STATE $G'\leftarrow$ Graph(CompositeNodes, Cut-Edges across partitions)
        \label{alg4-line:composite_graph}
		\RETURN \textit{UnionDP}($G'$)
		\label{alg4-line:recurse}
	\end{algorithmic}}
\end{algorithm}

\new{

\subsubsection{\textbf{Algorithm Description}}
\fullversion{In our implementation we use the Union-Find data structure, to efficiently build our partitions. We improve the $O(n^2)$ complexity IDP$_2$ has for finding the highest cost subtree, to $O(nlogn)$ for finding all relations sets which are given to the DP algorithm. The DP algorithm is still the major bottleneck in the overall optimization time at $O(3^k)$ complexity, so lowering the optimization time is mostly a factor of lowering the number of relations $k$ given to the DP algorithm.}

Algorithm~\ref{algo:uniondp} captures the pseudocode of {\uniondp}. The plans are built bottom up from the graph partitions recursively until the entire plan is constructed. If the number of relations is less than $k$, then we use {\mpdp}. The algorithm assigns weights to each edge based on a cost model (Line~\ref{alg4-line:assign_weights}), and the relations on either side of the edge are represented by \textit{leftRelSet} and \textit{rightRelSet}, respectively. Our algorithm uses the UnionFind data structure to maintain the partition information over relations, and for efficient \textit{find} and \textit{union} set operations. These sets are initialized to individual relations (Line~\ref{alg4-line:makeset}). 

Once the initialization is done, we traverse all the edges in increasing size of the sum of \textit{leftRelSet} and \textit{rightRelSet} relations  (Line~\ref{alg4-line:edges_itr}). Ties are broken by increasing weight of edges. \textit{leftRelSet} and \textit{rightRelSet} sets of the chosen edge will be unioned to the same set/partition if the size of their union is less or equal to $k$ (Line~\ref{alg4-line:union_condition}). At the end of this phase, called as \emph{partition phase}, size of all partitions are less than or equal to $k$. Then, all these partitions are individually optimized using {\mpdp} (Line~\ref{alg4-line:sub_optimize}). A new graph $G'$ is created with composite nodes (for each partition) and cut edges across partitions as its edge set  (Line~\ref{alg4-line:create-composite}). The above procedure is repeated over $G'$ until $|G'| <= k$ i.e. the size which {\mpdp} can handle efficiently (Line~\ref{alg4-line:end_condition_start}). This recursive idea helps {\uniondp} scale to 1000s of relations.

}
\section{{\mpdp}:  GPU Implementation}
\label{sec:gpu_algo}

In this section, we present details for GPU implementation  of {\mpdp}. 
GPUs provide a much higher degree of parallelism compared to multi-core CPUs. 
Recall that a DP-based optimization algorithm happen  in  several  levels,  finding the best plan at each level. Since  the  data  and  catalog  are  usually  resident  in  CPU, it calls functions  in  the  GPU  to  find  all  the  best  subplans  for every level $i$ repeatedly until the best plan is found. 

In our implementation, sets of relations (including adjacency lists of base relations) are represented using a fixed-width bitmap sets. The memo table is implemented using the fast Murmur3 hashing algorithm (a simple open-addressing hash table). 
Algorithm~\ref{algo:mpdp_gpu} shows the general workflow of  {\mpdp} 
on GPUs. \fullversion{Differently from the usual {\dpsub} algorithms, here sets are evaluated one level  at a time, like {\dpsize}, to avoid dependencies between sets processed in parallel.}
The memo table is initialized at Line~\ref{algo:generic_dpsub:memo}, then the memo table is
filled with the values derived from the base relations (Line~\ref{algo:generic_dpsub:fillmemo}) before starting the 
iterations (Line~\ref{algo:generic_dpsub:iter}). 
Each iteration is composed of the following steps:

\begin{algorithm}
	\setstretch{0.1}
	\caption{: MPDP on GPU}
	\label{algo:mpdp_gpu}	
	\begin{algorithmic}[1]
		\REQUIRE $QI$: Query Information
		\ENSURE Best (least cost) Plan
		\STATE $memo$ $\leftarrow$	Empty hashtable (key: relation id as bitmapset)
		\label{algo:generic_dpsub:memo}	
		\FORALL{$b \in QI.baseRelations$} \label{algo:generic_dpsub:fillmemo}	
		\STATE add $(b.id, b)$ to $memo$
		\ENDFOR
		\FOR{$i:=2$ to $QI.querySize$} \label{algo:generic_dpsub:iter}	
		\STATE $unrank$ all possible sets of size $i$ (Set $S_i$ in {\mpdp})
		\STATE $filter$ out not connected sets
		\STATE $evaluate$ all {\joinpairs} evaluated by {\mpdp}
		\STATE $prune$ retain the best one for each $\setS$ (optional)
		\STATE $scatter$ $(set,bestJoin(set))$ to $memo$
		\ENDFOR
		\RETURN best plan for query from $memo$
	\end{algorithmic}
\end{algorithm}

\emph{Unrank}.
All possible sets of relations of size $i$, corresponding to set $S_i$ in {\mpdp} (Line~\ref{alg3-line:Si} of Algorithm~\ref{algo:mpdp_generalization}), are unranked using a combinatorial schema as in~\cite{Meister2020}, and stored in a contiguous temporary memory allocation, which can be reused in successive iterations.

\emph{Filter}.
 All sets in $S_i$, that are not connected are filtered, thereby compacting the temporary array. This phase can be implemented using one of the many stream compaction algorithms for GPU, e.g. \texttt{thrust::remove}. 

\emph{Evaluate}.
The evaluation of {\joinpairs}, corresponding to Algorithm~\ref{algo:mpdp_generalization}, are performed with warp-level parallelism (one warp per set), using the parallel version of Find-Blocks ~\cite{bfs-bicc}, that finds all  blocks in $\setS$.
The warp first finds all the blocks for the given set, then each thread works on a different {\joinpair}. \new {Later, each thread unranks the blocks, checks validity and computes the cost. } 

\emph{Prune (optional)}.
In this step only the best {\joinpair} for each $\setS$ is kept. It can be easily implemented using any reduce-by-key algorithm.
\new{It is performed inside the warp in a parallel fashion using a classical warp reduction.}

\emph{Scatter}
All key-value pairs $(\setS, bestJoin(\setS))$ are saved in the memo table to be used in future iterations, \new{which is a parallel store on the GPU hash table}. 

Finally, the best plan is returned by fetching from the GPU memo table.  The final relation
is recursively fetched using its left and right join relations, building a join 
tree in CPU memory that can then be passed to {\postgres} as a schema to generate the 
final plan. 

\subsubsection*{\textbf{Enhancements}}

\new {The presented algorithm is inspired by the COMP-GPU algorithm from Meister et. al~\cite{Meister2020} and could be used to implement on GPU.
 The main difference with previous work is to propose the following enhancements:}

\emph{Reducing the number of global memory writes.}
Having a separate pruning phase, which runs in a separate kernel, requires storing the plans found by the threads in global memory, in order to perform the subsequent pruning step.
To remove this additional overhead, our implementation prunes the found plans in shared memory at the end of the \emph{evaluate} phase, so that only one write to global memory per warp is required, with the best plan for the set evaluated by the warp.
No separate pruning step is required.

\new{\emph{Avoiding 'If' branch divergence. }
In order to filter out invalid {\joinpairs},  the trivial solution would be to just use an \textbf{if} condition, but this would again cause branch divergence -- a major cause for performance degradation in GPUs.
Here, threads who find an invalid {\joinpairs} will stall until the other threads in the warp finish their execution, due to the SIMD nature of the GPU multi-processors.
We handle this issue by using Collaborative Context Collection~\cite{ccc} to prevent excessive in-warp divergence.} The basic idea is to defer work by stashing it in shared memory until there is enough work for all threads in the warp, either from enumerating new pairs or from the stash.

\fullversion{
\subsubsection{Improve unrank}
The \emph{unrank} phase can be improved by using a combinatorial schema which satisfies a lexicographical order (not shown here for sake of brevity) to unrank the first set of each thread, and then obtain subsequent sets using a O(1) algorithm for obtaining the next lexicographical bit permutation
\footnote{Specifically, the one showed in this webpage:
\url{https://graphics.stanford.edu/~seander/bithacks.html\#NextBitPermutation}}.
However, this has a negligible effect in case an expensive cost function
}

\section{Related Work}
\label{sec:relwork}

Join Order Optimization has been a well studied area, with over four decades of research work.  The importance of large queries with 1000 relations has been discussed in \cite{dieu20091, chen2009partial}. Further, they also showcase the inability of existing optimizers in handling such scenarios. We categorize the prior work into optimal algorithms and approximate (or heuristic) solutions. 

\subsubsection*{\textbf{Optimal Algorithms:}} 
One of the first approaches for join order optimization algorithm was {\dpsize}~\cite{selinger} -- which is currently used in many open-source and commercial databases, like {\postgres} and IBM DB2\cite{db2}, respectively.  {\dpsub}~\cite{dpsub}  uses a subset driven way of plan enumeration (detailed in Section~\ref{sec:bg}). 
Although {\dpsize} and {\dpsub}, depending on query join graph topology, evaluate a lot of unnecessary {\joinpairs}, they are parallelizable.
\fullversion{ However, on the positive side both these algorithms are highly parallelizable.  A detailed analysis of these two algorithms in terms of {\joinpairs} evaluated can be seen in ~\cite{dpccp}. Further,} 
Moerkette et al.~\cite{dpccp}  propose the {\dpccp} algorithm,  which uses a join graph based enumeration, and evaluates only {\ccpPairs}. {\dpccp} outperforms  both {\dpsize} and {\dpsub} while not considering cross products. However, due to dependencies between {\joinpairs} while plan enumeration, {\dpccp} is difficult to parallelize. Our work, {\mpdp} leverages massive parallelizability aspect of {\dpsub} while minimizing redundant evaluation of {\joinpairs} from {\dpccp}. 
A generalized version of {\dpccp}, \texttt{DPHyp}, has been developed by the same authors~\cite{dphyp}, which also consider hypergraphs, in order  to handle non-inner and anti-join predicates as well. Handling such cases is part of future work.

 With the rise of multicore architectures, parallel approaches for classical DP algorithms have been developed.
{\pdp}~\cite{HanKLLM08} discusses a CPU parallel version of {\dpsize}. Although, it scales up well,  its performance is hindered by the evaluation of a lot of invalid {\joinpairs}.
{\dpe}~\cite{HanL09} proposes a parallelization framework that can parallelize any DP algorithm, including {\dpccp}, based on a producer-consumer paradigm. {\joinpairs}  are enumerated in a producer thread and stored in a dependency-aware buffer, while consumers extract the {\joinpairs} from the buffer and compute their cost based on a cost model. Due to the sequential enumeration and the additional reordering step, its parallelizability is limited. 
Trummer et al. in \cite{trummer2016} proposes a novel plan space decomposition for join order optimization using shared-nothing architectures over large clusters. The main issue is that it is built on top of {\dpsize} that enumerates a lot of invalid {\joinpairs} in realistic settings, and hence does not scale to large number of relations. 
More recently, Meister et al.~\cite{Meister2020} proposed GPU versions of {\dpsize} and {\dpsub} algorithms. 

\new{We use \cite{Meister2020}'s basic GPU implementation structure which has unrank,  filter, evaluate, prune and scatter phases. We propose enhancements over \cite{Meister2020} such as  reducing global writes by not having a separate pruning phase, and avoiding branch divergence due to 'if' condition by using Collaborative Context Collection.}

\subsubsection*{\textbf{Heuristic Solutions:}} Due to the NP-hard nature of join order optimization, there have been approaches that use heuristic solutions. Some techniques look at only a limited search space of query plans.
\texttt{IKKBZ}~\cite{ikkbz1,ikkbz2} limits the search space to left-deep join trees only.
  Similarly, Trummer and Koch formulate the join order optimization problem as a Mixed Integer Linear Programming (MILP) problem~\cite{milp}. 

Techniques such as \texttt{GOO}~\cite{Fegaras1998} and \texttt{min-sel}~\cite{min-sel} greedily choose the best subplans to find the query plan.

 \texttt{IDP}~\cite{Kossmann2000}, introduces a new class of algorithms, called Iterative Dynamic Programming which we have discussed in Section~\ref{sec:idp2}.  
More recently, Neumann et al.~\cite{Neumann2018} proposed a new technique to reduce the search space, called linearized DP ({\lindp}), which runs a DP algorithm to optimize the best left-deep plan found by IKKBZ. In order to handle large queries, they propose an adaptive optimization technique. Their technique employs {\dpccp} for small queries (<14 tables), linearized DP for medium queries (between 14 and 100), and IDP$_2$ with linearized DP for large queries (>100 tables).

There are also randomized algorithms proposed based on Simulated Annealing and Iterative Improvement~\cite{randomized1}, Genetic Algorithms~\cite{bennett1991genetic,HorngLK94}, Random Sampling~\cite{Waas2000}. Some recent work such as \cite{marcus2021bao, krishnan2018learning} use machine learning techniques for query optimization. 
Primary issue with these approaches are that either they do not scale well for large join queries considered in this work, or produce low quality solutions \cite{Neumann2018}.

\section{Experimental Results}
\label{sec:expt}

In this section, we discuss the experimental evaluation of both {\mpdp}'s optimal algorithm and its heuristic solution on different  workloads, while comparing it against corresponding state-of-the-art algorithms. \new{We note that the focus of this paper is on optimization of large queries, i.e, join of 10 or more relations.}

\subsection{Experimental Setup}
\label{sec:exp_setup}
For our experiments, we use a server with dual Intel Xeon E5-2650L v3 CPU, with each CPU having 12 cores and 24 threads with 755GB of RAM. We use a Nvidia GTX 1080 GPU for running GPU based join algorithms. \new{For the experiments in Section~\ref{sec:expt:aws}, we use Amazon AWS~\cite{aws}.}

We have implemented all the join order optimization techniques ({\mpdp} and all baselines) in the {\postgres} 12 \cite{postGreSql} engine. 
Since implementing these algorithms (plus GPU-specific implementation) require code changes to the optimizer module, we cannot provide experimental results on commercial databases. 

\emph{Cost Model:} The cost model used by a query optimizer plays an important role in determining the optimization time. While  recent works such as \cite{Neumann2018}, have used a cost model based on output size of different operators, i.e. $c_{out}$, we use a more realistic cost model which is close to the one used by {\postgres}. 
For the suite of queries considered in this paper,  our cost model  returns nearly the same cost as {\postgres} (within 5\% in the worst case). 
\footnote{\new{We do not use the original cost model of {\postgres} since we only consider inner equi-joins while {\postgres} cost model covers a lot more cases (e.g. outer joins, inequality joins as well as degree of parallelism > 1). Using the exact {\postgres} cost model would require us to rewrite from scratch over 20 thousand lines of cost model code from {\postgres} for GPU based execution, which is beyond the scope of the paper.
\fullversion{GPUs have different execution models compared to CPUs, therefore CPU data structures are not as efficient in GPU and thus functions for GPUs need to be rewritten and tuned.}}
}
 
\subsection{Optimal Algorithm Evaluation}
\label{sec:expt_mpdp}
Our goal for the experiments in this section is to see how our optimal {\mpdp} algorithm (both CPU and GPU based implementations) performs compared to other optimal DP algorithms. Since all algorithms produce the optimal plan, we just compare their optimization times. We use both synthetic and real-world workloads for the evaluation.  Note that the size of the dataset does not make much difference to the optimization time.

We use the following baselines in our experiments.
 
\begin{itemize}[noitemsep,topsep=0pt,leftmargin=*]
    \item {\texttt{Postgres (1CPU)}}: {\dpsize} based join ordering implemented by {\postgres} running on 1 CPU core.  
    \item {\dpccp\texttt{ (1CPU)}}: State-of-the-art CPU Sequential DP algorithm, {\dpccp}~\cite{dpccp}, running on 1 CPU core. 
    \item {\texttt{DPE (24CPU)}}: State-of-the-art CPU parallel algorithm, {\dpe}. We use the parallel version of {\dpccp} \cite{HanL09} running on 24 cores.
    \item {\texttt{DPSub (GPU)}}: State-of-the-art GPU based DP algorithm~\cite{Meister2020} using {\dpsub}. The COMB-GPU version from \cite{Meister2020} is used.

    \item {\texttt{DPSize (GPU)}}: Another state-of-the-art GPU based DP algorithm~\cite{Meister2020} using {\dpsize}. The H+F-GPU version is used here.
\end{itemize}
Other techniques such as sequential or parallel versions of {\dpsize} and {\dpsub} on CPU run much slower than their GPU variants, and are hence omitted to make the graphs easier to read. 
Similarly, {\dpe} performs better than {\pdp}~\cite{HanKLLM08}, we skip {\pdp} as well. 

We set a timeout of 1 minute for the total optimization time, and report the average optimization times across several queries of each query size. 
\new{For joins with less than 10  relations {\mpdp} (GPU) does not perform that well  because of data transfers cost between CPU and GPU for every level in the DP lattice. In terms of absolute values, for such small queries the optimization times are usually less than 10ms for all techniques including {\mpdp} (GPU). }

\fullversion{
\begin{figure}
\centering
\includegraphics[width=0.45\textwidth,keepaspectratio=true]{result/chain.pdf}
\caption{Optimization times on chain graph}
\label{fig:chain}
\end{figure}
}

\begin{figure}
\centering
\includegraphics[width=0.35\textwidth,keepaspectratio=true]{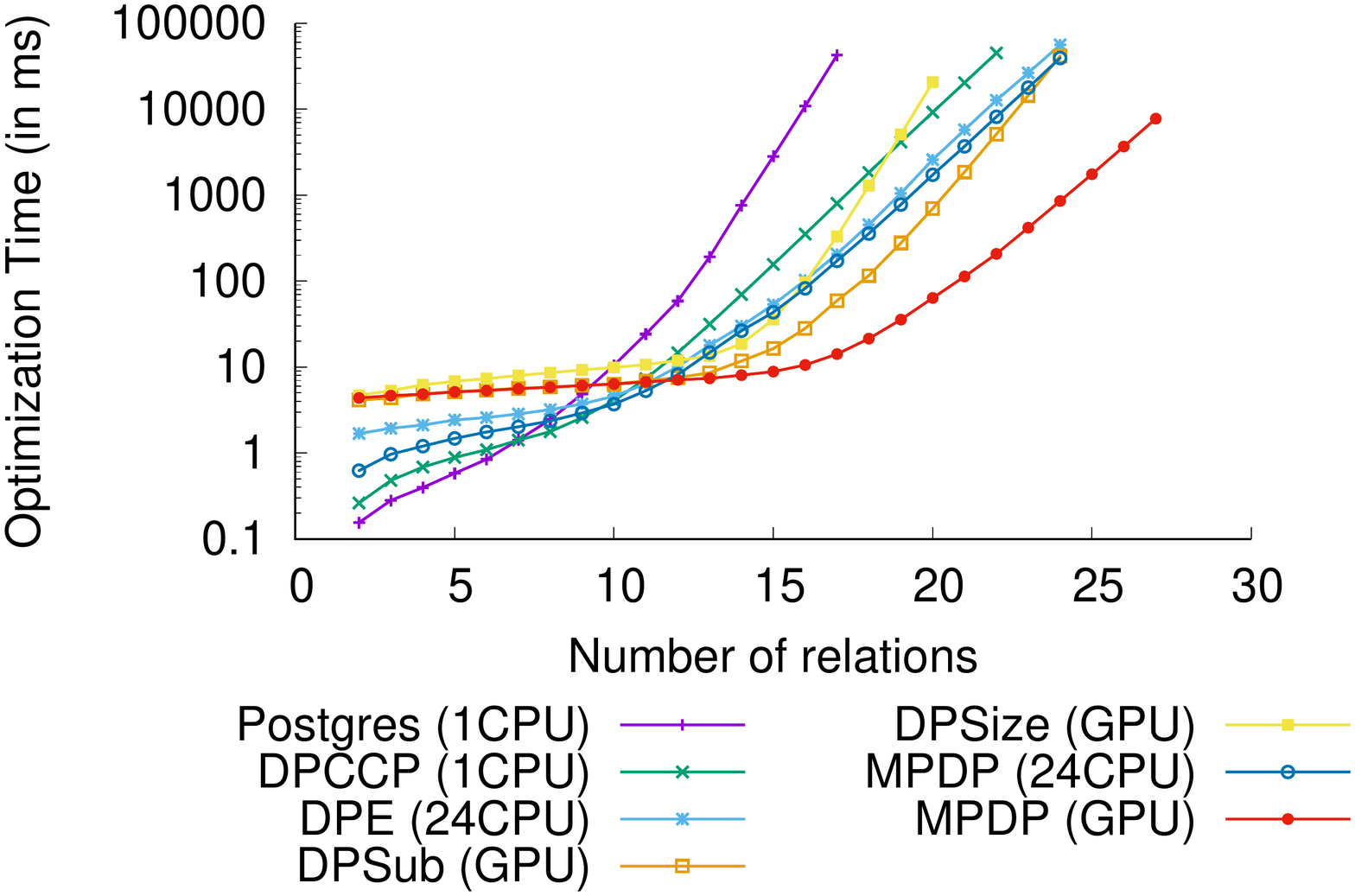}
\caption{Optimization times on star graph
}
\label{fig:star}
\end{figure}

\begin{figure}
\centering
\includegraphics[width=0.35\textwidth,keepaspectratio=true]{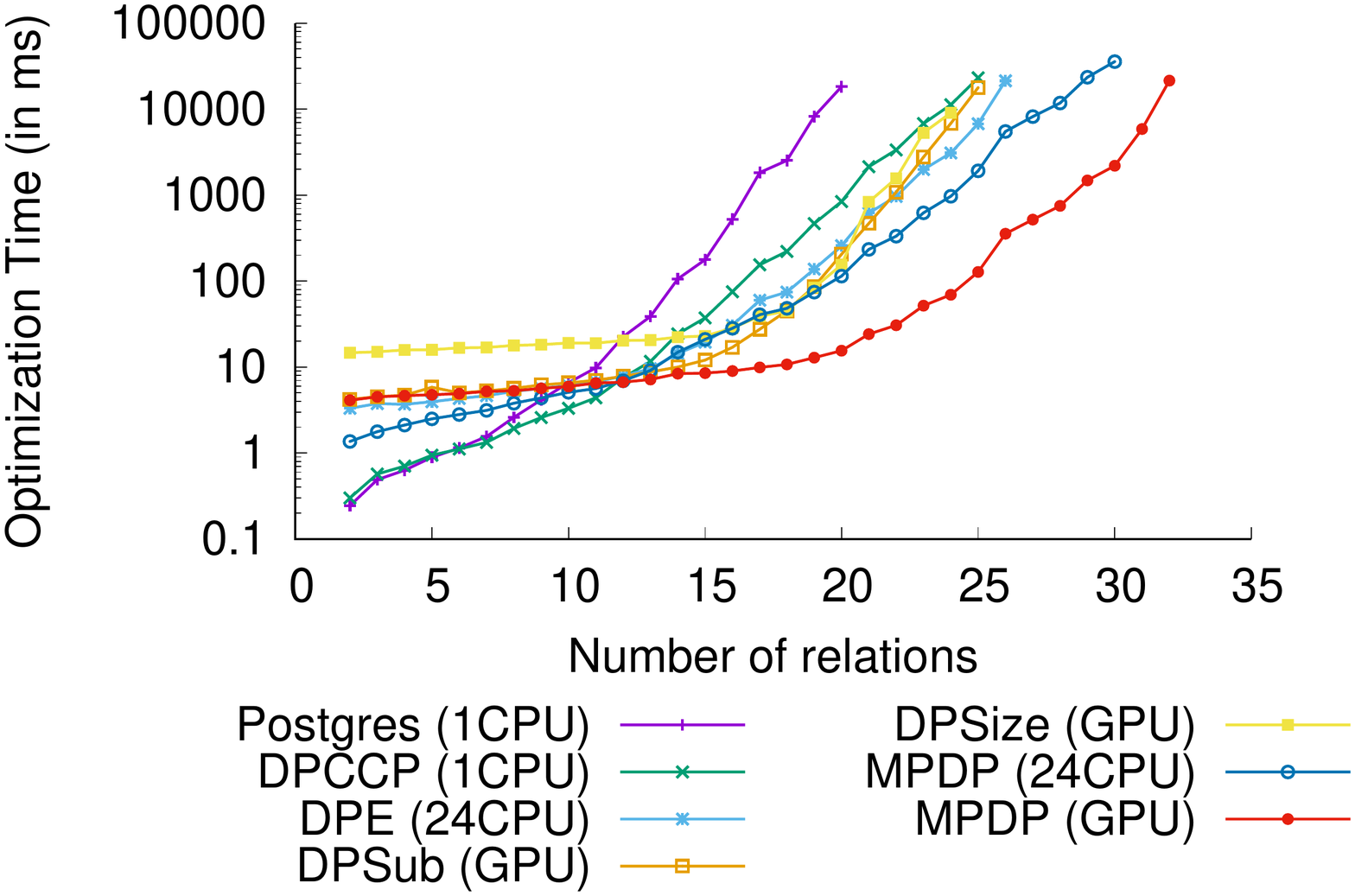}
\caption{Optimization times on snowflake graph
}
\label{fig:snowflake}
\end{figure}

\begin{figure}[t]
\centering
\includegraphics[width=0.35\textwidth,keepaspectratio=true]{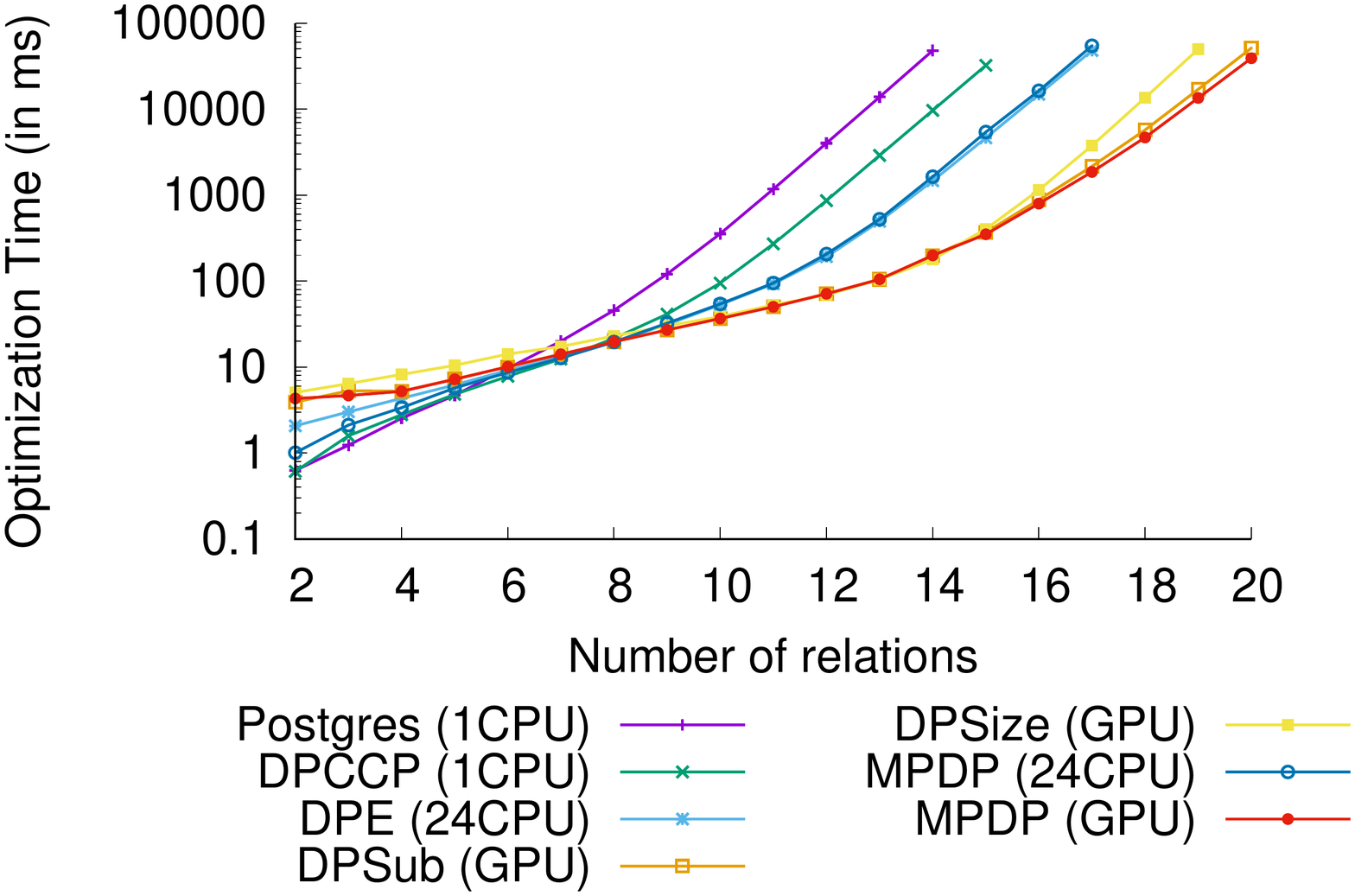}
\caption{Optimization times on clique graph
}
\label{fig:clique}
\end{figure}

\fullversion{
\begin{figure*}
\centering
\includegraphics[width=0.9\textwidth,keepaspectratio=true]{result/synthetic.pdf}
\caption{Optimization times on clique graph}
\label{fig:syn}
\end{figure*}
}

\subsubsection{\textbf{Synthetic Workload}}
In this section, we present our evaluation results on synthetic workloads. 
We generate queries with different type of join graphs and with different number of relations.\footnote{The equivalence classes introduced because of joins in the given query may change the join graph since they introduce implicit predicates. Our join graph also take into account the equivalence classes.} We consider the following types of join graphs: 
\begin{enumerate}[noitemsep,topsep=0pt,leftmargin=*]
    \item \textbf{Star join graph}: In this type of join graph there is a single fact relation to which other dimension relations join. \fullversion{This is a very common case with star-schema data warehouses queries that involve primary-key foreign key joins.} 
    \item \textbf{Snowflake join graph}: The join graph for these queries creates a snowflake pattern. The maximum depth we use is  4. 
    \item \textbf{Clique join graph} In this type of join graph all relations have joins to all other relations and the join graph is a clique. All {\joinpairs} in this case are valid {\joinpairs} (equivalently, capturing the cross join scenario), and hence join ordering for these graphs are more expensive to compute. 
\end{enumerate}

\new{Star and snowflake join graphs are very common scenario with analytics on data warehouse, while cliques are typically not as realistic which showcase the performance  when cross joins are considered.} For chain join graphs, only polynomial number of valid {\joinpairs} are present. The search space for join order optimization in such queries are much smaller, and we found that {\dpccp}, {\dpe} and {\mpdp (GPU)} were able to optimize 30 relation joins within 100 ms. Hence, we do not consider them in our evaluation. The optimization times for each of the above workload is presented next.

\fullversion{Figure~\ref{fig:chain} shows the optimization time for chain join graphs.  As can be seen from the graph, DPE and DPSize perform best in this case. For a chain query graph, MPDP is not able to efficiently parallelize the execution. However MPDP(GPU), even though it does relatively poorly in this case, can perform 35 relation join optimization in under 10 sec.}

\emph{Star Join Graph}:
 The optimization times for star join graphs are shown in Figure~\ref{fig:star}. The X-axis denotes the number of relations in the join query, while the Y-axis shows the optimization times in millisecs. 
As can be seen from the figure, {\mpdp} (GPU) outperforms all the baselines by at least an order of magnitude beyond 21 relations, and can generate the plan within 2s for 25 relations. Moreover, {\mpdp} (GPU) scales well with number of relations.

At 25 relations, {\mpdp} (GPU) is 17x faster than its 24CPU version because of the parallelism offered by GPUs. 
{\mpdp} (GPU) is 20x faster compared to {\dpsub} (GPU),  due to evaluating 2805 times fewer {\joinpairs}. The gap is even bigger compared to {\dpsize} (GPU) as it evaluates 12024x fewer  {\joinpairs} at 20 relations.
Also, {\mpdp} (GPU) has a speedup of at least 3 orders of magnitude over \texttt{postgres (1CPU)} from 16-relations, and a speedup of at least 2 orders of magnitude over \texttt{DPCCP (1CPU)} from 18-relations. {\dpe} (24CPU) takes over 70x longer to optimize queries with join of 23 relations compared to {\mpdp} (GPU).

\fullversion{and 230x faster than 1CPU\footnote{Not reported in the plots for sake of clarity.}.}

\emph{Snowflake Join Graph}: 
For snowflake schema, the results are shown in Figure~\ref{fig:snowflake}. 
For snowflake graphs too, {\mpdp} (GPU) outperforms all baselines  by at least an order of magnitude beyond 22 relations. In this case, {\mpdp} (GPU) can perform join optimization for 30 relations under 3s while other techniques timeout at 26 tables. Similar to star queries, the most noticeable difference is between {\mpdp (GPU)} and {\dpsub (GPU)}, with {\mpdp} (GPU) being 56x faster at 27 tables. This increase is due to the more efficient enumeration of {\joinpairs} by {\mpdp}.

\emph{Clique Join Graph}: 
The optimization times for clique join graphs is shown in Figure~\ref{fig:clique}. 
Since these join graphs are fully connected, pruning the search space is not feasible. Massive parallelization, however, still play a crucial role, with all GPU based algorithms outperforming all the CPU ones. Specifically, even {\dpsize} (GPU)  overtakes all the CPU algorithms, which was not the case in the previous two workloads. Furthermore, {\mpdp} (GPU) performs the best, closely followed by {\dpsub (GPU)}. However, {\dpsize (GPU)} is 3x slower at 19 tables, because {\dpsize} checks additional overlapping pairs, which are not enumerated by {\mpdp} (GPU) and {\dpsub} (GPU).

\fullversion{
, since {\mpdp} improvements have no effect in cliques, only providing additional overhead. 
Overall, it may seem that other techniques are closer to {\mpdp}, but it's important to note that they can be compared up to only 17 tables. {\mpdp (GPU)} speedup would be higher with more relations. 
Furthermore, differently from the other topologies, in clique DPSize (GPU) closely follows at around 2x, due to the total time being dominated by the joins, instead of the checks, as previously. }

\begin{figure}
\centering
\includegraphics[width=0.35\textwidth,keepaspectratio=true]{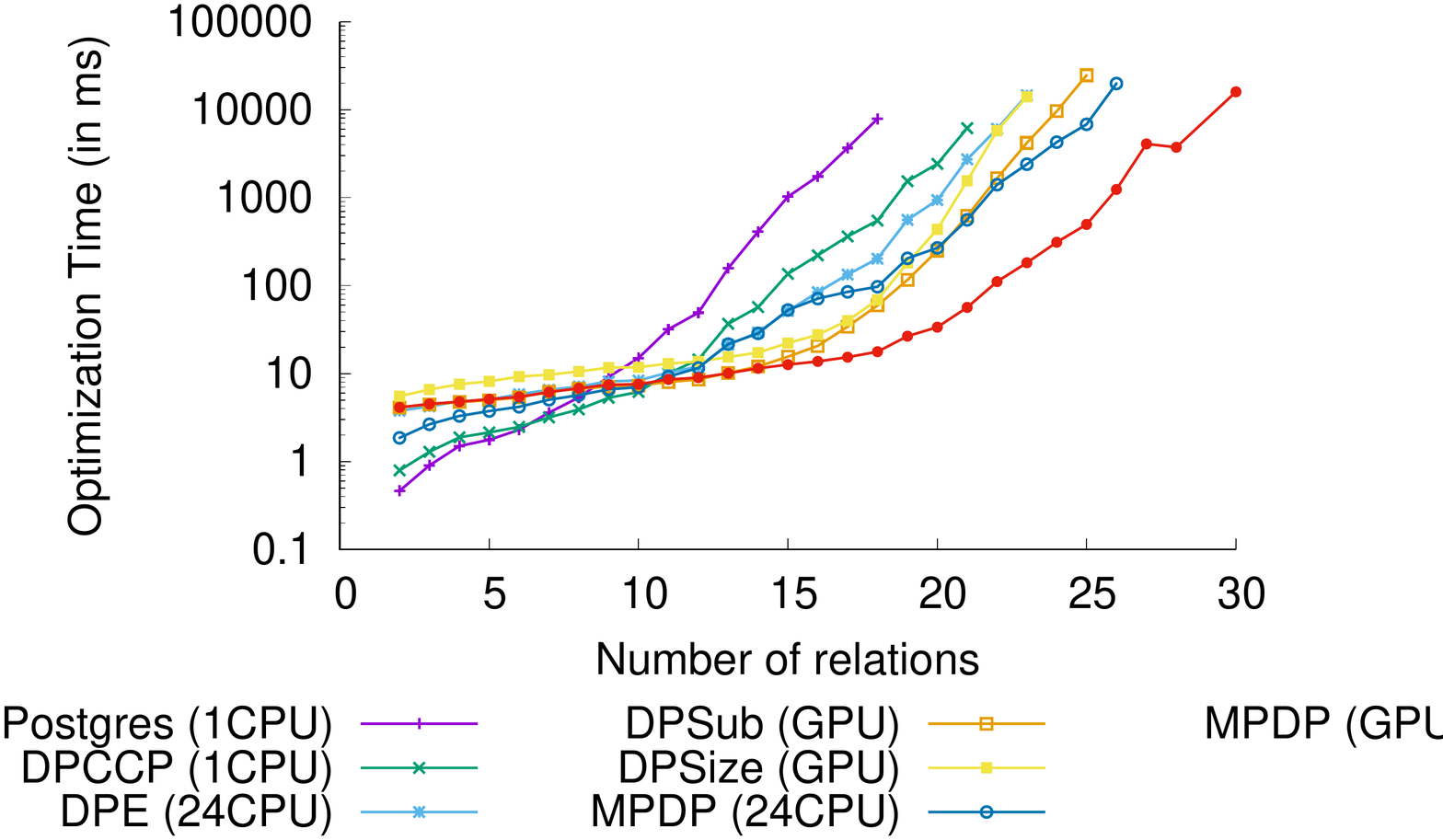}
\caption{Optimization times on MusicBrainz queries 
}
\label{fig:musicbrainz}
\end{figure}

\newcolumntype{?}{!{\vrule width 0.8pt}}

\begin{table*}
\caption{Heuristic cost comparison for snowflake schema}
\label{tab:approximate}
\new{
\resizebox{\textwidth}{!}{
\begin{tabular}{?p{2.61cm}?r|r?r|r?r|r?r|r?r|r? r|r?r|r?r|r?r|r?r|r?r|r? r|r?} \Xhline{2\arrayrulewidth}
\multirow{2}{1.4cm}{Technique/ \# tables} & \multicolumn{2}{c?}{30}  & \multicolumn{2}{c?}{40}  & \multicolumn{2}{c?}{50} &\multicolumn{2}{c?}{60}  & \multicolumn{2}{c?}{80}& \multicolumn{2}{c?}{100}  &\multicolumn{2}{c?}{200} & \multicolumn{2}{c?}{400} & \multicolumn{2}{c?}{500} & \multicolumn{2}{c?}{600} & \multicolumn{2}{c?}{800} & \multicolumn{2}{c?}{1000} \\\cline{2-25}
{}&  avg &  95\% &  avg &  95\% &  avg &  95\% &  avg &  95\% &  avg &  95\% &  avg &  95\% &  avg &  95\% &  avg &  95\% &  avg &  95\% &  avg &  95\% &  avg &  95\% &  avg &  95\%\\\Xhline{2\arrayrulewidth}
{GE-QO} & {1.9} & {2.3} & {2.1} & {2.5} & {2.2} & {2.8} & {2.4} & {2.8} & {2.4} & {3.0} & {2.5} & {3.1} & {3.1} & {3.8} & - & - & - & - & - & - & -& -& -& -\\\hline
{GOO} & {1.5} & {1.9} & {1.6} & {2.1} & {1.6} & {2.2} & {1.7} & {2.3} & {1.7} & {2.3} & {1.8} & {2.4} & {2.1} & {2.6} & {2.2} & {2.7} & {2.2} & {2.7} & 2.3 & 2.7 & 2.2 & 2.8 & {2.1} & {2.5}\\\hline
{LinDP} & {1.6} & {2.2} & {2.0} & {2.8} & {2.3} & {3.2} & {2.7} & {3.9} & {3.4} & {4.6} & {4.2} & {5.7} & {4.6} & {6.8} & {4.4} & {7.0} & {4.4} & {7.0}  & 4.0 & 6.6 & 3.2 & 5.5 & {3.0} & {5.3}\\\hline
{IKKBZ} & {3.5} & {4.5} & {4.4} & {5.6} & {5.4} & {7.0} & {6.3} & {8.0} & {8.2} & {10.0} & {10.1} & {12.5} & {18.2} & {21.8} & {32.1} & {37.6} & {38.4} & {44.2} & -& - & -& -& -& -\\\hline
{IDP$_2$-MPDP~(15)} & {1.2} & {1.5} & {1.3} & {1.7} & {1.4} & {1.8} & {1.4} & {1.8} & {1.5} & {1.9} & {1.5} & {1.9} & {1.7} & {2.2} & {1.7} & {2.2} & {1.7} & {2.2}  & 1.8 & 2.2  & 1.9 & 2.3 & {1.9} & {2.2}\\\hline
{IDP$_2$-MPDP~(25)} & {1.1} &  {1.5} & {1.2} & {1.6} & {1.3} & {1.8} & {1.4} & {1.8} & {1.4} & {1.8} & {1.5} & {1.9} & {1.6} & {2.0} & {1.7} & {2.2} & {1.7} & {2.0} & {1.7} & {2.2} & 1.7 & 2.2 & {1.7} & {2.2}\\\hline
{UnionDP-MPDP~(15)} & {\textbf{1.0}} & {\textbf{1.0}} & {\textbf{1.0}} & {\textbf{1.0}} & {\textbf{1.0}} & {\textbf{1.0}} & {\textbf{1.0}} & {\textbf{1.0}} & {\textbf{1.0}} & {\textbf{1.0}} & {\textbf{1.0}} & {\textbf{1.0}} & {\textbf{1.0}} & {\textbf{1.0}} & {\textbf{1.0}} & {\textbf{1.0}} & {\textbf{1.0}} & {\textbf{1.0}} & {\textbf{1.0}} & {\textbf{1.0}}  & {\textbf{1.0}} & {\textbf{1.0}} & {\textbf{1.0}} & {\textbf{1.0}}\\\Xhline{2\arrayrulewidth}
\end{tabular}
}}
\end{table*}

\begin{table*}
\caption{Heuristic cost comparison for star schema}
\label{tab:approximate:star}
\new{
\resizebox{0.85\textwidth}{!}{
\begin{tabular}{?p{2.61cm}?r|r?r|r?r|r?r|r?r|r? r|r?r|r?r|r?r|r?r|r?r|r? } \Xhline{2\arrayrulewidth}
\multirow{2}{1.4cm}{Technique/ \# tables} & \multicolumn{2}{c?}{30}  & \multicolumn{2}{c?}{40} & \multicolumn{2}{c?}{50}  & \multicolumn{2}{c?}{60}  & \multicolumn{2}{c?}{80} &\multicolumn{2}{c?}{100}  & \multicolumn{2}{c?}{200}& \multicolumn{2}{c?}{300}  &\multicolumn{2}{c?}{400} & \multicolumn{2}{c?}{500}& \multicolumn{2}{c?}{600} \\\cline{2-23}
{}&  avg &  95\% &  avg &  95\% &  avg &  95\% &  avg &  95\% &  avg &  95\% &  avg &  95\% &  avg &  95\% &  avg &  95\% &  avg &  95\% &  avg &  95\% &  avg &  95\% \\\Xhline{2\arrayrulewidth}
{GE-QO}  & 1.2 & 1.5 & 1.3 & 1.6 & 1.4 & 1.9 & 1.4 & 1.7 & 1.4 & 1.8 & 1.3 & 1.7 & 1.3 & 1.6 &- &- & - & - &-  & - &-  &- \\\hline
{GOO} & 1.4 & 2.3 & 1.6 & 2.6 & 1.7 & 2.8 & 1.7 & 2.8 & 1.7 & 2.9 & 1.7 & 2.9 & 1.5 & 2.6 & 1.6 & 2.9 & 1.7 & 2.9 & 1.6 & 2.9 & 1.6 & 2.9\\\hline
{LinDP} & 1.4 & 2.3 & 1.6 & 2.6 & 1.6 & 2.8 & 1.7 & 2.8 & 1.6 & 2.9 & 1.7 & 3.0 & 1.5 & 2.6 & 1.6 & 2.9 & 1.7 & 2.9 & 1.6 & 2.8 & 1.6 & 2.9\\\hline
{IKKBZ} & 1.4 & 2.3 & 1.6 & 2.6 & 1.7 & 2.8 & 1.7 & 2.8 & 1.7 & 2.9 & 1.7 & 2.9 & 1.5 & 2.6 & 1.6 & 2.9 & 1.7 & 2.9 & 1.6 & 2.8 &- &-\\\hline
{IDP$_2$-MPDP~(15)} &\textbf{ 1.0} & \textbf{1.0}  &\textbf{1.0} & \textbf{1.0} & \textbf{1.0} & \textbf{1.0}  & \textbf{1.0} & \textbf{1.0} & \textbf{1.0} & \textbf{1.0} & \textbf{1.0} & \textbf{1.0} & \textbf{1.0} & \textbf{1.0} & \textbf{1.0} & \textbf{1.0} & \textbf{1.0} & \textbf{1.0} & \textbf{1.0} & \textbf{1.0} & \textbf{1.0} & \textbf{1.0}\\\hline
{IDP$_2$-MPDP~(25)} &\textbf{ 1.0} & \textbf{1.0}  &\textbf{1.0} & \textbf{1.0} & \textbf{1.0} & \textbf{1.0}  & \textbf{1.0} & \textbf{1.0} & \textbf{1.0} & \textbf{1.0} & \textbf{1.0} & \textbf{1.0} & -&  -& -& - &-  & -&-  & - & - &- \\\hline
{UnionDP-MPDP~(15)} &\textbf{1.0} & \textbf{1.0}  & \textbf{1.0} & \textbf{1.0} & \textbf{1.0} & \textbf{1.0} & \textbf{1.0} & \textbf{1.0} & \textbf{1.0} & \textbf{1.0} & \textbf{1.0} & \textbf{1.0} & \textbf{1.0} & \textbf{1.0} & \textbf{1.0} & \textbf{1.0} & \textbf{1.0} & \textbf{1.0} & \textbf{1.0} & \textbf{1.0} &- &-\\
\Xhline{2\arrayrulewidth}
\end{tabular}
}}
\end{table*}

\subsubsection{\textbf{Real-world Workload}}
\label{sec:expt:real}
We evaluate our techniques on a real world  {\musicbrainz} dataset~\cite{musicbrainz}. 
This database, consisting of 56 tables, include information about artists, release groups, releases, recordings, works, and labels in the music industry. Since we do not have access  to query logs, we generate our own queries. We  only consider widely used primary key - foreign key joins. We pick a relation at random and then do a random walk on the graph till we get the required number of rels, $n$. For any given number of relation, $n$, we generate 15 such queries and  report its average values. 
Note that the generated queries can contain cycles.
\eat{We generate two types of join graphs on the  MusicBrainz dataset, one with cycles and one without cycles. To generate join graphs without cycles we take the join graph with cycles and extract a minimum spanning tree version of it provided by Postgres. }

The results for the join order optimization on MusicBrainz dataset is shown in Figure~\ref{fig:musicbrainz} and is formatted similar to the graphs for the synthetic datasets. 
Again, {\mpdp (GPU)} outperforms  all the baselines at least by an order of magnitude beyond 24 relations. {\mpdp (GPU)} can optimize a 30 join query within 45s. 
For 26 tables, the {\mpdp (GPU)} is 14x faster than its CPU counterpart (24 threads)
and 19x faster than {\dpsub (GPU)}. 
The outperformance over {\mpdp} (24CPU) is due to the increased parallelism provided by GPUs, while over {\dpsub (GPU)} is because of evaluating fewer invalid {\joinpairs}. {Also {\mpdp} (GPU) is 80 times faster than \texttt{DPE (24CPU)} for joins with 23 rels.}

{\dpsize}-based algorithms do not perform well due to checking too many overlapping pairs, with {\dpsize (GPU)} doing better than \texttt{DPE (24CPU)} for sizes between 13 and 22, only due to the higher computational power at disposal, before its optimization time becomes dominated by checking invalid {\joinpairs}.

\subsubsection{Comparison of Optimization and Execution Times}
\label{sec:opt-exec-times}
\begin{figure}[t]
\begin{center}
 \subfloat[Primary Key - Foreign Key Joins]{\includegraphics[width=0.2\textwidth, height = 2 cm]{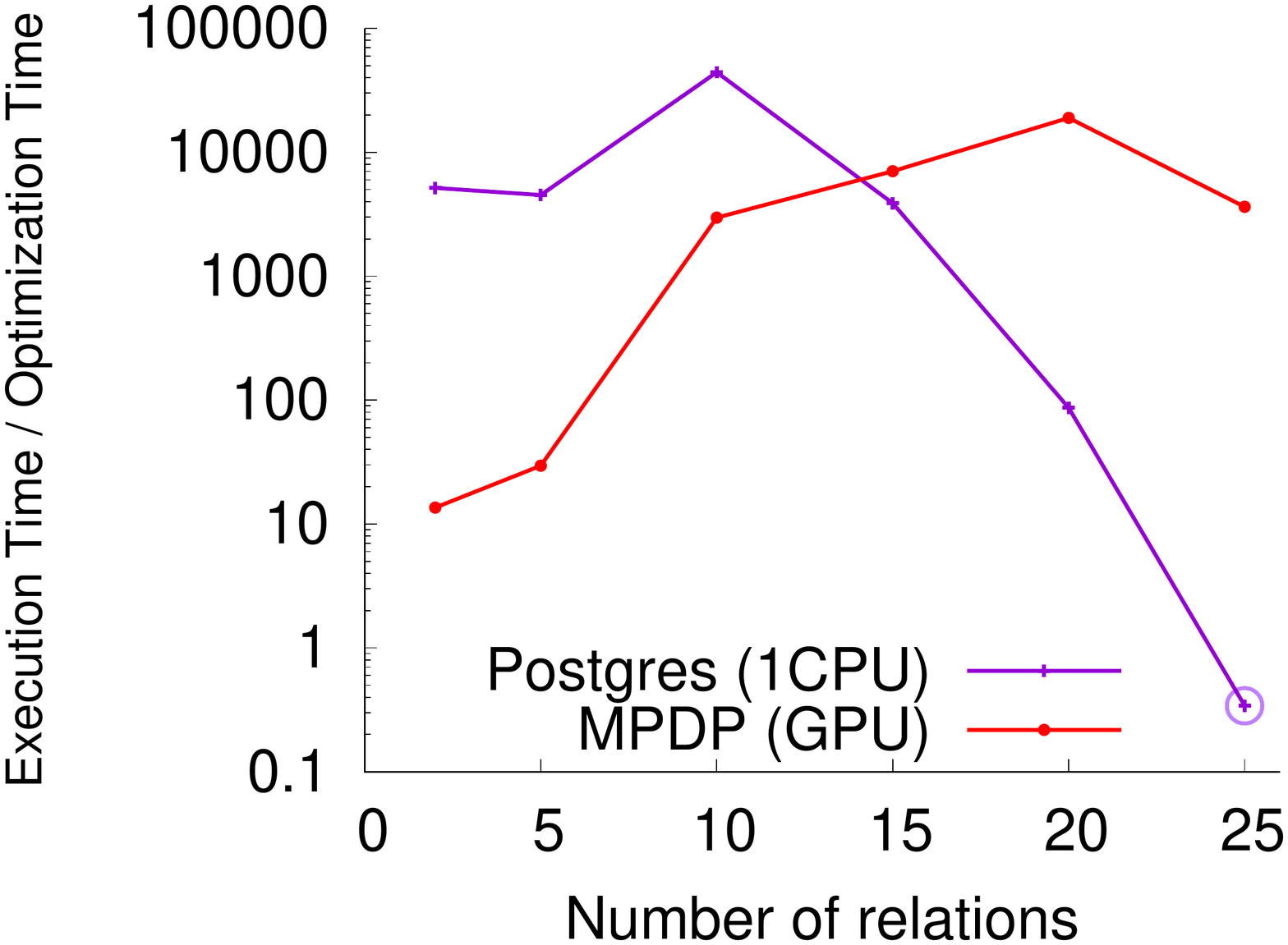}}
\quad
 \subfloat[Non Primary Key - Foreign Key Joins]{\includegraphics[width=0.2\textwidth, height = 2 cm]{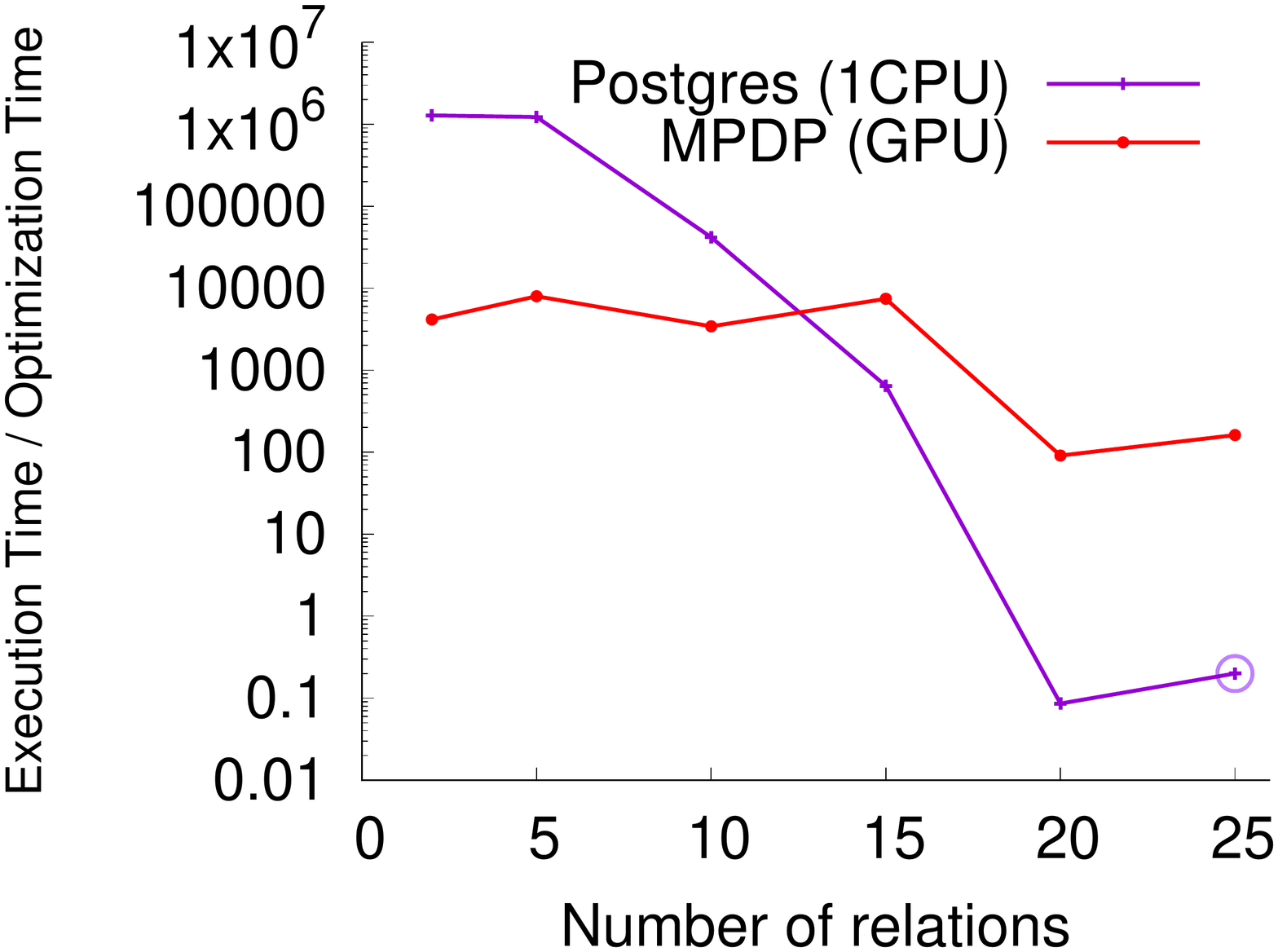}}
 \caption{Ratio of Execution and Optimization Times}
 \label{fig:opt-exec}
 \end{center}
\end{figure}

\revise{To evaluate the significance of optimization time for large queries, we compare the execution time and the optimization time for queries on the MusicBrainz dataset. Note that the execution times also depend on the size of the dataset,  query predicates and the execution environment (eg: number of machines). Recall that our experimental setup is limited to a single machine. Primary key-foreign key (PK-FK) joins and non PF-FK joins have different execution time characteristics. Hence we use non PK-FK joins also for this experiment. 
Figure~\ref{fig:opt-exec} captures the ratio of execution vs optimization times averaged over 10 queries for each relation size. 

The results show that with the {\postgres} optimizer, the optimization time is a significant portion of the total query processing time (i.e. optimization + execution) for large queries. 
For both PK-FK and non-PK-FK scenarios, for 25 relations, the {\postgres} optimization timed out for all queries (with a timeout of 3 hours). We conservatively set its optimization time to the timeout value. In this case, the query execution, given the optimal plan, finishes in a fraction of the timeout value. The same is not true with {\mpdp} (GPU) since the optimization time is much less as compared to {\postgres}. 
Thus, the experiment demonstrates that  the savings in optimization time achieved by {\mpdp} is highly beneficial, especially for joins with large number of relations. }

\new{
\subsubsection{JOB benchmark}
	\label{sec:expt:job}
We now present the optimization time for queries from the Join order benchmark (JOB)~\cite{Leis2016}. While JOB does not have many queries with large number of joins, (with the largest query involving a 17 relation join) it is the only benchmark we could find with at least some queries with large number of joins. Further, as pointed out in \cite{Neumann2018}, JOB mainly stresses on quality of cardinality estimation. 
	
The results on JOB for various optimization techniques are shown in Figure~\ref{fig:job}. 
{\mpdp} (GPU) starts outperforming others from around 12 relations, while closely followed by {\dpsub} (GPU). The performance gap between {\mpdp} and {\dpsub} increases with more relations, with {\mpdp} (GPU) being 2.3 times faster at 17 relations.

\begin{figure}
\centering
\includegraphics[width=0.35\textwidth,keepaspectratio=true]{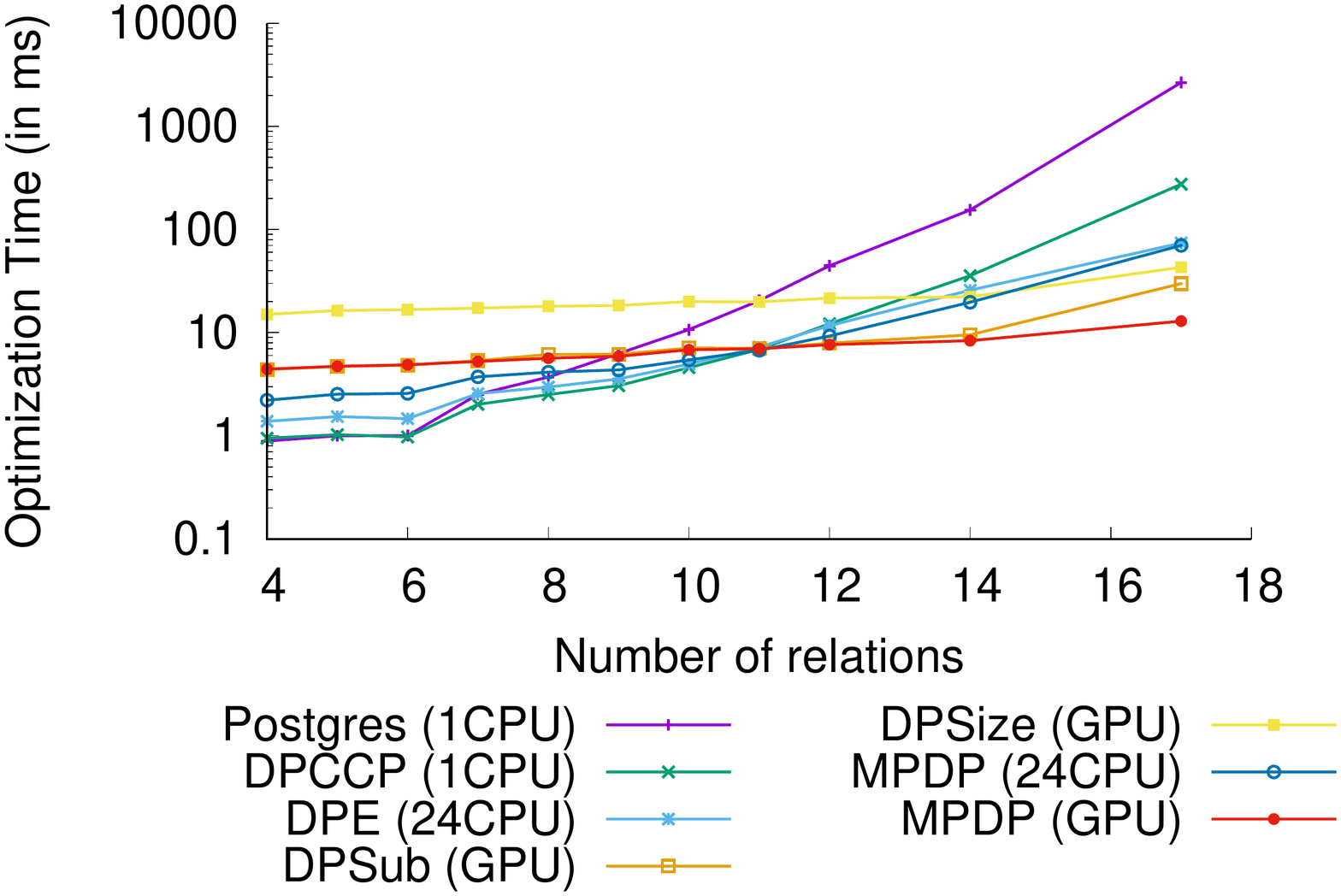}
\caption{JOB query optimization times}
\label{fig:job}
\end{figure}
}

\new{
\subsubsection{Impact of GPU Implementation Enhancements}
\label{sec:gpu_enhancements}

There are primarily two enhancements over \cite{Meister2020} (Section~\ref{sec:gpu_algo}):
    
    1) \emph{Reducing global memory writes through Kernel fusion}, whose  improvement depends on complexity of cost function, and yield up to 40\% improvement on {\mpdp}. 
    
    2) \emph{Collaborate Context Collection} (CCC), whose impact depends on graph topology, achieve up to 3X improvement with {\mpdp}.  

}

\subsection{Heuristic Solution Evaluation}
\label{sec:expt:heuristic}
The optimization time using {\mpdp}, although much better than other techniques, rises exponentially with the number of relations.  In order to evaluate larger join sizes than what would be feasible using our techniques, we presented  heuristics in Section~\ref{sec:approx}. In this experiment, we evaluate our IDP$_2$ based heuristic and {\uniondp}  with other heuristics based on 
\fullversion{ two criteria: (a) the optimization time to get the join order; (b) }
the quality of the plan produced (i.e. total cost) using our {\postgres}-like cost model.

Note that we do not compare the actual query execution times but only compare the costs since the actual execution time may be different from what is estimated due to errors in cost and cardinality models.  
Handling those errors are beyond the scope of this paper. 

In this set of experiments, we compare our optimization techniques with the following heuristic techniques:
\begin{itemize}[noitemsep,topsep=0pt,leftmargin=*]
    \item \texttt{GE-QO}: Genetic algorithm based optimization used in \linebreak PostgreSQL~\cite{Whitley1994}. We use the default parameters. 
    \item \texttt{GOO}~\cite{Fegaras1998}: Greedy Operator Order which uses the resulting join relation size to greedily pick the best join at each step.
    \item \texttt{IKKBZ}~\cite{ikkbz1,ikkbz2}: Technique that finds the best left-deep tree, which is also used in linearized DP. It uses the $C_{out}$ cost function to estimate the best left-deep join order.
    \item \texttt{LinDP}~\cite{Neumann2018}: Adaptive optimization technique, which chooses among {\dpccp}, linearized DP and IDP$_2$ using GOO and the linearized DP depending on query size.
    The linearized DP is a novel technique that optimizes the left-deep plan found by IKKBZ in polynomial time.
    The default thresholds presented in the original paper have been used.
\end{itemize}

\new{For all IDP2 variants, we use GOO (Greedy Operator Ordering) for the heuristic step. }
\new{We evaluated IDP$_2$ for $k =5, 10, 15, 20, 25$. Due to space limitations, we just present its median (i.e. 15) and maximum value (i.e. 25). As we increase $k$ the plan quality increases. For instance, IDP$_2$-MPDP for 30 rels, has normalized cost values 1.4, 1.27, 1.23, 1.17 and 1.14, for $k=5, 10, 15, 20, 25$, respectively. Further, higher values of $k$ (i.e., > 25) can be chosen with larger timeouts. For {\uniondp}, we use $k=15$ since plan quality were similar with  $k=25$, while running much faster.}

\new{
We use the snowflake and star synthetic schema to evaluate the approximation heuristics. We also ran experiments on clique join graphs. The snowflake schema is the most likely one to be used in analytical queries for such large queries.
We only consider primary key - foreign key joins. For the star schema, we generate queries with selections so that different join orders would result in different costs of intermediate joins.}
In order to get statistically significant results to compare the costs, we generate 100 queries for each join relation size that we consider in the heuristic optimization techniques. We also set the optimization timeout to 1 min. 
We do not use the {\musicbrainz} database as it has only 56 tables.

The relative execution cost, for the snowflake schema, as given by the cost model is shown in Table~\ref{tab:approximate}. For each query, we set the cost of the best plan found by any algorithm to 1 and find the relative cost of other techniques with respect to the best plan. We show the average relative cost and the 95th percentile of the relative cost measured in this manner across 100 queries. The best relative costs for each case are marked in bold.

\new{
As shown in the table, \texttt{UnionDP-MPDP~(15)}  provide the best query plans across all join sizes. This is because {\uniondp} can easily create partitions by removing single expensive edges.  \texttt{IDP$_2$-MPDP~(25)} performs the second best. We also see that there is some advantage in using a bigger value of $k$ for \texttt{IDP$_2$-MPDP~(25)}. 
The genetic optimization used by {\postgres} produces plans that are on average 2-4x more expensive than the best-found plan, also, it timeouts after 200 rels.
\texttt{UnionDP-MPDP~(15)} produced plans that were 1.5-2.3x times cheaper than \texttt{GOO} on average. 
  
\texttt{IKKBZ} scales poorly with number of rels, going up to a factor of 38.4 (on average) at 500 tables. This is \eat{probably} because \texttt{IKKBZ} only considers left-deep plans.  \texttt{LinDP} achieves 1.6-4.6x worse plans in average compared to \texttt{UnionDP-MPDP~(15)} while the 95-percentile goes up to 7x. 
}

\new{For star join graphs, the results are shown in Table~\ref{tab:approximate:star}. Both our techniques, produce the cheapest query plans which are much better than other techniques. 
For instance at 100 relations, GE-QO, GOO, LinDP and IKKBZ  are 1.3-1.7x costlier than that of \texttt{IDP$_2$-MPDP} and \texttt{UnionDP-MPDP} on average. Compared to the case of snowflake schema, IKKBZ does not perform as bad  since  the optimal join order also falls in IKKBZ's search space of left-deep plans.  

\fullversion{\texttt{IDP$_2$-MPDP~(15)}
and \texttt{IDP$_2$-MPDP~(25)} marginally outperform \texttt{UnionDP-MPDP~(15)} by 10\%. Here, {\uniondp} is not able to find partitions without removing many nodes since all dimension tables join to a single fact table. \texttt{IDP$_2$-MPDP~(25)} times out at 100 relations while \texttt{IDP$_2$-MPDP~(15)} times out at 500 relations.
}

For large uncommon clique join graphs, we summarize the evaluation results in the interest of space. 
All techniques time out much earlier (compared to snowflake) with {\lindp} at 70 rels, while GOO, IDP$_2$-MPDP and UnionDP-MPDP at around 100 rels. Here, IDP$_2$-MPDP has the best performance. However, GOO produce up to 2x lesser quality plans compared to IDP$_2$-MPDP. {\uniondp} does not perform that well because it has too many edges which creates an issue for balancing partition size and maximizing cut edges.  
 } 

\fullversion{
\begin{figure}
\centering
\includegraphics[width=0.45\textwidth,keepaspectratio=true]{result/approximate.pdf}
\caption{Snowflake Graph Heuristic: Optimization Time}
\label{fig:approximate}
\end{figure}
}

\subsection{Scalability on CPU}
\label{sec:exp:scalability}
Now we will see how our solution scales with CPU threads.  We use a 20 relation query on MusicBrainz and vary the number of threads  from 1 to 24. The results are similar for other relation sizes.
The scalability factor also depends on the cost function complexity (also captured~\cite{Meister2017}). We use the cost function as described earlier  to get results representative of a real world scenario. 

The scalability results 
 are shown in Figure~\ref{fig:scalability}.
The X-axis represents the total number of threads used, while the Y-axis represents the speedup with respect to a single thread execution of the corresponding algorithm. 
{\mpdp} scales much better  compared to DPE. This is because DPE cannot parallelize the candidate join pairs enumeration, but can only evaluate join costs in parallel. Further, {\mpdp} scales sub-linearly beyond 6 threads since the CPU caches get swapped out when many join pairs are evaluated in parallel.

\begin{figure}
\centering
\includegraphics[width=0.3\textwidth,keepaspectratio=true]{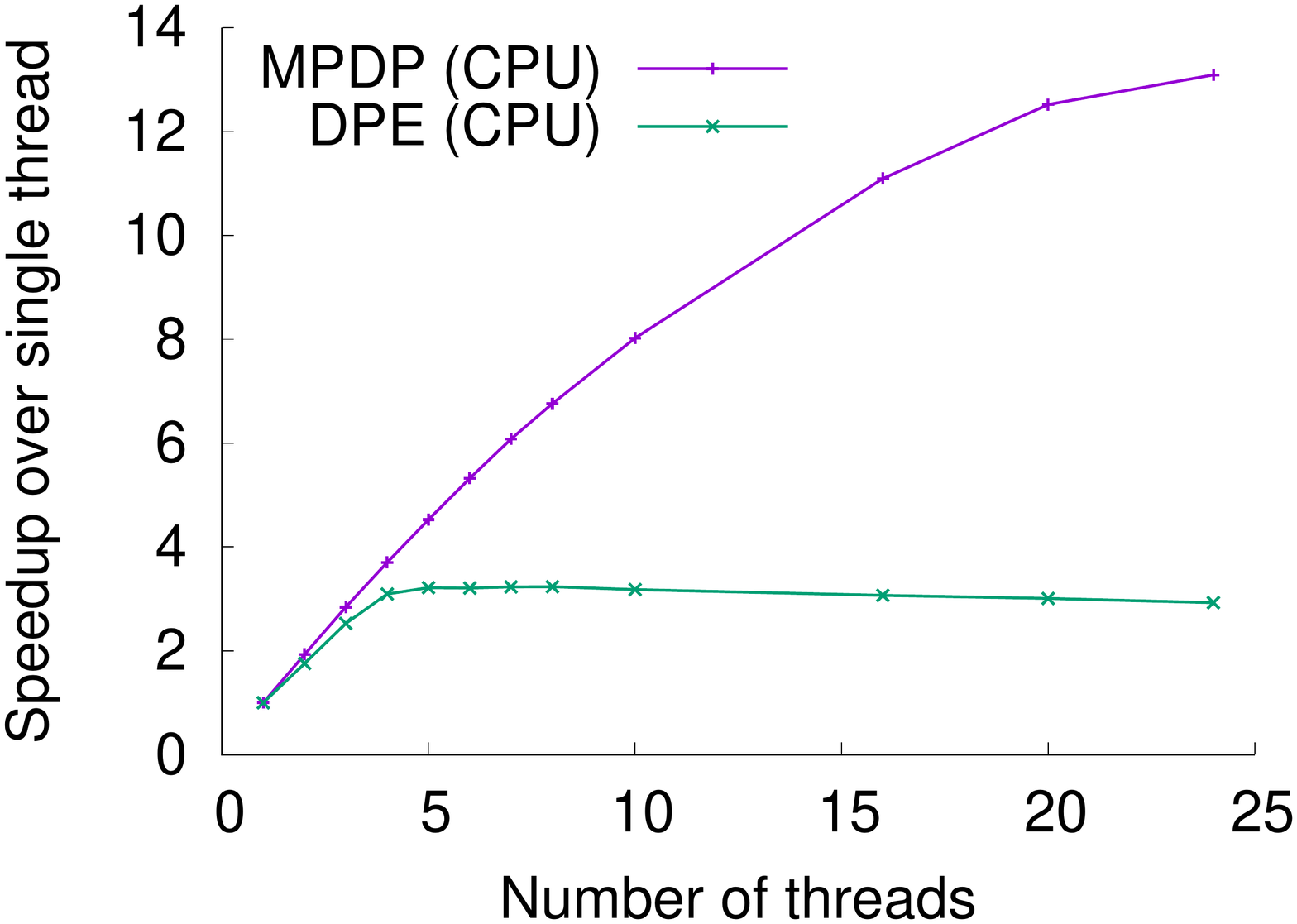}
\caption{CPU Scalability on MusicBrainz. }
\label{fig:scalability}
\end{figure}

\fullversion{
\begin{figure}
\centering
\includegraphics[width=0.35\textwidth,keepaspectratio=true]{result/scalability_star.pdf}
\caption{CPU Scalability on star queries. MPDP (CPU) scales much better than DPE (CPU)}
\label{fig:scalability_star}
\end{figure}
}

\new{
\subsection{Optimization cost comparison} 
\label{sec:expt:aws}
While the above experiments provides analysis between different CPU and GPU implementations of different algorithms, the monetary cost of using different  hardware may be different. In this experiment, we compare the cost of optimization of these techniques while using Amazon AWS cloud instances. 

Since the CPU algorithms do not scale linearly with large number of cores, we experimented with different AWS size instances and picked the one that is the most cost effective. For the single threaded CPU algorithms, {\dpccp} and Postgres optimizer, we used a  \textit{c5.large} instance which has 2 vCPU cores and 4 GiB of memory. For {\dpe} and MPDP (CPU) we used a \textit{c5.xlarge} instance which has 4 vCPU cores and 8 GiB of memory. For GPU based algorithms, we used a \textit{g4dn.xlarge} instance which has an NVIDIA T4 GPU.

The result of the experiment is shown in Figure~\ref{fig:aws}. The X-axis shows the number of relations, while the Y-axis shows the cost of optimization for a single query in U.S. cents. We obtained the cost by multiplying the time taken for the optimization with the amount paid for running the instance per unit time. For smaller queries, {\postgres}  and {\dpccp} are cheaper. However, for larger queries (beyond 15 relations), {\mpdp}\texttt{ (GPU)} turns out to be the cheapest. For instance, {\mpdp} is around an order of magnitude cheaper compared to next best {\dpe} from 23 relations.  

\begin{figure}
\centering
\includegraphics[width=0.35\textwidth,keepaspectratio=true]{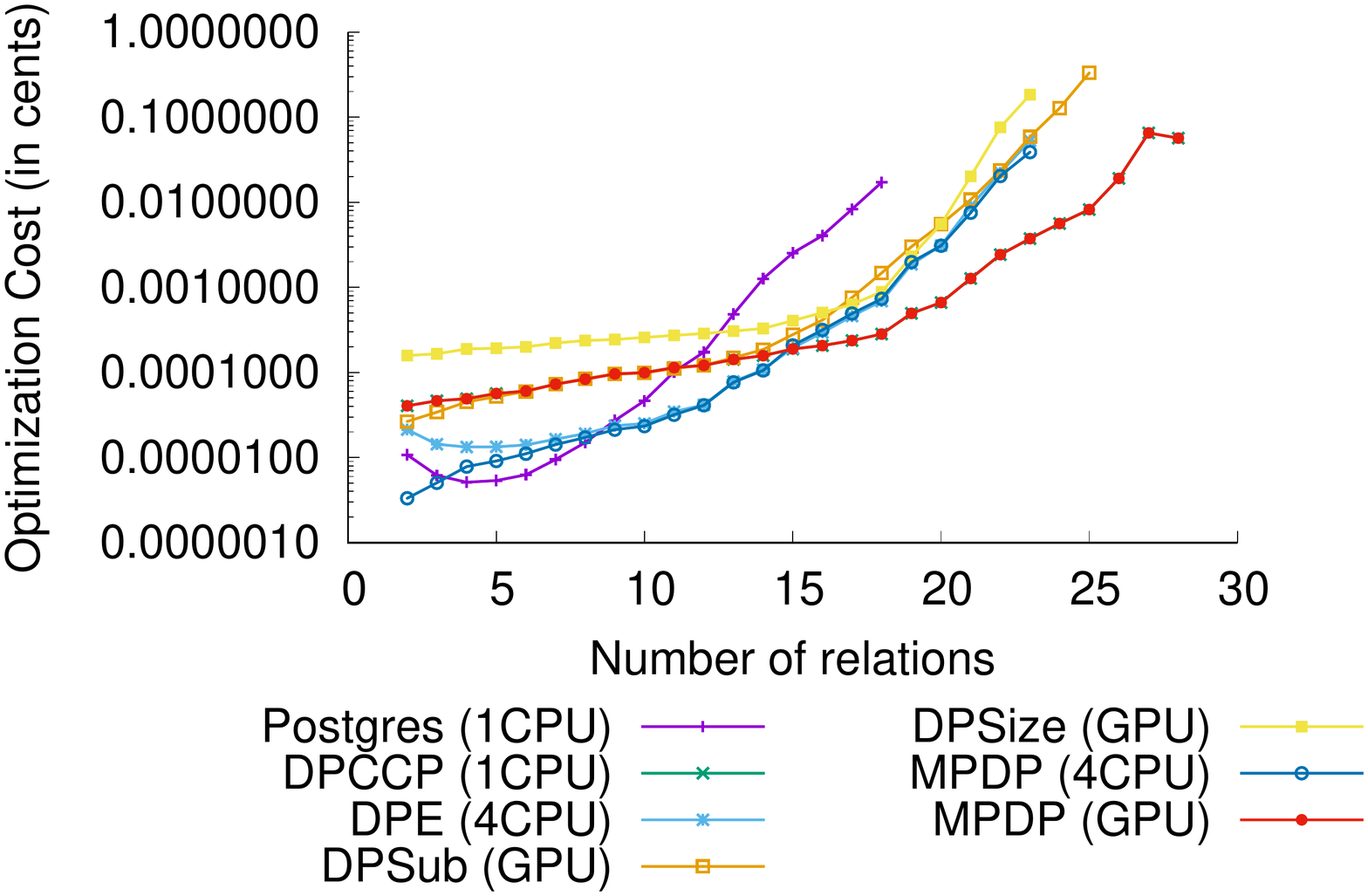}
\caption{Cost of optimization on AWS}
\label{fig:aws}
\end{figure}
}
\section{Conclusions}
\label{sec:conc}
In this paper, we described techniques for join order optimization for queries with large number of joins. Our query optimization technique is capable of running in parallel, while significantly pruning the search space  and can be efficiently implemented on GPUs too. Our experiments, in case of exact scenario, show that our techniques significantly outperform other state-of-the-art techniques  in terms of query optimization time. Our heuristic solutions allow us to efficiently explore a larger search space for very large join queries  (eg: 1000 rels), thereby allowing us to find plans with much better costs compared to state-of-the-art heuristic techniques. Areas of future work may include, using our optimization framework for scenarios with increased optimization search space such as  cloud analytics, graph analytics and bigdata systems.

\bibliographystyle{ACM-Reference-Format}
\bibliography{references}

\pagebreak
\end{document}